\documentclass[preprint]{elsarticle}

\usepackage{amsmath}
\usepackage{amsfonts}
\usepackage{amssymb}
\usepackage{amsthm}		
\usepackage{eucal}		
\usepackage{bm}                 
\usepackage{mathtools}          

\usepackage{hyperref}
\newcommand{\oeisseqnum}[1]{\href{https://oeis.org/#1}{\rm \underline{#1}}}

\renewcommand\ge\geqslant
\renewcommand\le\leqslant
\renewcommand\propto\varpropto

\newcommand{\Hermiteprob}{\mathop{{}\it He}\nolimits}

\newcommand{\normord}[1]{\mathopen{:}\,\mathrel{#1}\,\mathclose{:}}

\newcommand{\HS}{\mathop{{}\it HS}\nolimits}

\newtheorem{theorem}{Theorem}[section]
\newtheorem{corollary}[theorem]{Corollary}
\newtheorem{proposition}[theorem]{Proposition}

\theoremstyle{definition}
\newtheorem{definition}[theorem]{Definition}
\newtheorem{example}[theorem]{Example}

\theoremstyle{remark}

\newtheorem{remarkaftertheorem}{Remark}[theorem]
\newtheorem*{remark*}{Remark}

\numberwithin{equation}{section}

\newcommand{\adag}{{a^\dag}}

\newcommand{\defeq}{\mathrel{:=}}
\newcommand{\eqdef}{\mathrel{=:}}

\makeatletter
\newenvironment{sizeequation}[1]{%
  \skip@=\baselineskip
  #1%
  \baselineskip=\skip@
  \equation
}{\endequation \ignorespacesafterend}

\newenvironment{sizealignat}[2]{%
  \skip@=\baselineskip
  #1%
  \baselineskip=\skip@
  \alignat
  #2%
}{\endalignat \ignorespacesafterend}  
  
\makeatother

\begin{document}


\begin{frontmatter}
\title{Sheffer polynomials and the s-ordering of exponential boson operators}
\author{Robert S. Maier}
\ead{rsm@math.arizona.edu}
\address{Depts.\ of Mathematics and Physics, University of Arizona, Tucson,
AZ 85721, USA}

\begin{abstract}
  The $s$-ordered form of any product of single-mode boson creation
  and annihilation operators, containing only a single annihilator, is
  computed explicitly.  The $s$\nobreakdash-ordering concept
  originated in quantum optics, but subsumes normal, symmetric (Weyl),
  and anti-normal ordering for any two operators satisfying a
  canonical commutation relation.  Because the
  $s$\nobreakdash-ordering map can be viewed as producing a function
  of a complex variable, its inverse is a quantization map that takes
  such ``classical'' functions to quantum operators.  The explicit
  $s$\nobreakdash-ordered expressions are derived with the aid of a
  parametric family of Sheffer polynomial sequences (or~equivalently a
  parametric exponential Riordan array of polynomial coefficients),
  called the Hsu--Shiue family.  To yield orderings interpolating
  between normal and anti-normal, this family must be extended.
\end{abstract}
\end{frontmatter}


\section{Introduction}
\label{sec:intro}

The Heisenberg--Weyl algebra~$\mathcal{W}$, generated by ladder
operators $a,\adag$ with canonical commutation relation $[a,\adag]=1$,
appears in the modeling of a quantum harmonic oscillator, a single
quantized radiation mode of the electromagnetic field, or in~general,
any quantized, single-mode bosonic field.

Even in a single-mode theory, when calculating quantum-mechanical
expectation values of observables, it is important to be able to
express the corresponding operators in~$\mathcal W$, and operators
in~general (such as the density operator representing the current
state, and also non-hermitian operators), in various ordered ways.  An
example is the exponential operator ${\rm e}^{\lambda (\adag a)}$, an
exponentiated version of the number operator~$\adag a$, which when
$\lambda<0$ appears as the density operator for a thermal state.

In the then new field of quantum optics, Cahill and
Glauber~\cite{Cahill69} derived as identities
in~$\mathcal{W}[[\lambda]]$, the space of formal power series
in~$\lambda$ with coefficients in~$\mathcal{W}$, the representations
\begin{equation}
\label{eq:CH1}
  {\rm e}^{\lambda(\adag a)}
    =
    \left\{
    \begin{aligned}
      &\normord{\exp\left[ \left({\rm e}^\lambda-1\right)\adag a \right]}_N,
      \\
      &\normord{\frac2{1+{\rm e}^\lambda}
        \exp \left [ \frac{ 2({\rm e}^\lambda-1)}{1+{\rm e}^\lambda} \adag a \right]}_W,
      \\
      & \normord{{\rm e}^{-\lambda}\exp\left[ \left(1-{\rm e}^{-\lambda}\right)\adag a \right]}_A.
    \end{aligned}
    \right.
\end{equation}
The unary double dot operations $\normord{\dots}_N$,
$\normord{\dots}_W$, and $\normord{\dots}_A$ are quantization maps
that enforce normal, Weyl (i.e., symmetric), and anti-normal ordering
on each monomial (in~$\lambda$ and $a,\adag$) in the series expansions
of the functions on which they act.  The first of these identities is
the McCoy--Schwinger formula, which was previously known.  Within each
$\normord{\dots}$, the formal symbols $a,\adag$ are taken to commute
and are often written as~$\alpha,\alpha^*$; so each of the three
functions is a ``c\nobreakdash-number'' or ``classical''
representation of the quantum operator ${\rm e}^{\lambda(\adag a)}$.

They also introduced a concept of $s$\nobreakdash-ordering
$\normord{\dots}_s$, which generalizes the three more familiar types:
it involves a free parameter $s\in\mathbb{C}$ (usually, $s\in[-1,1]$),
and when $s=-1,0,+1$, $s$\nobreakdash-ordering reduces to N-,~W-, and
A\nobreakdash-ordering.  They derived the identity
\begin{equation}
\label{eq:CH2}
  {\rm e}^{\lambda(\adag a)}
  =
  \normord{
  \frac{2}{1+{\rm e}^\lambda - s(1-{\rm e}^\lambda)}
        \exp \left [ \frac{ 2({\rm e}^\lambda-1)}{1+{\rm e}^\lambda-s(1-{\rm e}^\lambda)} \adag a \right] }_s,
\end{equation}
which subsumes the preceding three, and has the quantization
\begin{sizeequation}{\small}
\label{eq:invCH2}
  \normord{{\rm e}^{\lambda(\adag a)}}_s =
    \frac2{2-\lambda-s\lambda}
    \left(
    \frac{2+\lambda-s\lambda}{2-\lambda-s\lambda}
    \right)^{\adag a}
    \!= \left[1+\tfrac12(1-s)\lambda\right]^{\adag a}
    \left[1-\tfrac12(1+s)\lambda\right]^{-a\adag}    
\end{sizeequation}
as its inverse.

The $s$\nobreakdash-ordering concept has proved useful in quantum
mechanics on phase space, and specifically in the quasiprobability
approach to quantum optics~\cite{Schleich2001,Leonhardt97} and other
fields of quantum physics~\cite{Sperling2020}.  A~fundamental result
is Sudarshan's optical equivalence theorem, which makes the quantum
theory of light formally resemble a classical probabilistic
theory~\cite{Klauder68}.  Suppose a density operator~$\hat\rho$ for a
state has \emph{anti-normal} ordered representation
$\normord{F^{(A)}(\alpha,\alpha^*)}_A$ (which means that $F^{(A)}$~is
some function on the complex $\alpha$-plane, real due to hermiticity,
and that anti-normal quantization of~$F^{(A)}$ yields~$\hat\rho$).
Then, $F^{(A)}$~can be viewed as a quasiprobability density on
$\mathbb{C}\ni\alpha$, from which quantum-mechanical expectations of
\emph{normal} ordered operators, with respect to the state, can be
computed by integration in a formally classical way.

This $F^{(A)}$, which is obtained (by setting $s=\nobreak+1$) from a
general $s$-parametrized representation of~$\hat\rho$, which is
denoted by $F^{(s)}$, may not be non-negative on the
$\alpha$\nobreakdash-plane; in which case it cannot be interpreted as
a classical probability density, and the state is said to be
non-classical.  But non-classicality can be quantified: as $s$~is
lowered from~$+1$ toward $0$ and~$-1$, regions of negativity (if~any)
of~$F^{(s)}$ and its marginals shrink at some rate and
vanish~\cite{Orlowski93}.

Quantum optics models in which the time evolution of
$s$\nobreakdash-parametrized quasiprobability densities has been
studied include the widely applicable Jaynes--Cummings model of a
two-level atom (a~qubit) interacting with a single light mode of a
cavity~\cite{Daeubler92}, and the Drummond--Walls model of dispersive
optical bistability~\cite{Vogel89}.  Plots of densities evolving with
time, at multiple alternative values of~$s$, reveal much about model
dynamics.

In this paper, we go beyond previous work on the foundations of
$s$\nobreakdash-ordering by finding closed-form, analytic expressions
of $s$\nobreakdash-ordered representations denoted by~$F^{(s)}$ (i.e.,
$s$\nobreakdash-dependent functions on the complex $\alpha$-plane),
for certain single-mode exponential boson operators of the form ${\rm
  e}^{\lambda\omega}$, where $\omega\in\mathcal{W}$.  That is, for
each~$s$ we find an~$F^{(s)}$ for which the
$s$\nobreakdash-quantization $\normord{F^{(s)}}_s$ equals~${\rm
  e}^{\lambda\omega}$.

The case most amenable to an analytic treatment is when $\omega$~is a
monomial or ``boson string'': a~product of $a$'s and~$\adag$'s, i.e.,
a word over~$\{a,\adag\}$.  The ``single annihilator'' word
$\omega=\adag^L a\adag^R$ will be treated, where $e=L+\nobreak
R-\nobreak1\ge0$ is the excess of $\adag$'s over~$a$'s.  The
Cahill--Glauber result, eq.~(\ref{eq:CH2}) above, is the
$s$\nobreakdash-representation of ${\rm e}^{\lambda\omega}$ when
$(L,R)=(1,0)$ and $\omega=\adag a$.

The chief tool is the theory of Sheffer polynomials
\cite{Roman84,Costabile2022}, or equivalently the theory of
exponential Riordan arrays (infinite triangular arrays of polynomial
coefficients) \cite{Shapiro91,Barry2016,Shapiro2022}.  A~Sheffer
polynomial sequence is a sequence of polynomials $s_n(t)$ with
$s_0(t)\equiv1$ and $\deg s_n=n$, $n\ge0$, which can be constructed in
any of several ways, including the use of differential ladder
operators that increment and decrement~$n$.  By extending previous
work of Lang~\cite{Lang2000,Lang2007} and B{\l}asiak and collaborators
\cite{Blasiak2005,Blasiak2006}, we first derive normal and anti-normal
representations that subsume the first and third identities
in~(\ref{eq:CH1}), namely
\begin{align}
&
\begin{aligned}
{\rm e}^{\lambda(\adag^L a \adag^R)}
&=\normord{\mathcal{HS}(\adag a, \lambda\adag^e;\, -e,1,R)}_N\\
&= \normord{(1-e\lambda\adag^e)^{-R/e} \exp\left\{ \left[\bigl(1-e\lambda\adag^e\bigr)^{-1/e}-1\right] \adag a\right\}}_N
\end{aligned}
\label{eq:3}
\\
\shortintertext{and}
&
\begin{aligned}
{\rm e}^{\lambda(\adag^L a \adag^R)}
&=\normord{\mathcal{HS}(a \adag, \lambda\adag^e;\,e,-1,-L)}_A\\
&= \normord{\exp\left\{ - a\adag \left[\bigl(1+e\lambda\adag^e\bigr)^{-1/e}-1\right] \right\}   (1+e\lambda\adag^e)^{-L/e} }_A,
\end{aligned}
\label{eq:4}
\end{align}
where the $e=0$ case is handled by taking the formal $e\to0$ limit.
Here, $\mathcal{HS}(t,z;A,B,r)$ is the exponential generating function
(EGF) of a certain parametric family of Sheffer polynomial sequences,
named after its originators Hsu and Shiue, and parametrized
by~$A,B,r$.  As mentioned, $a,\adag$ inside each $\normord{\dots}$
could be written as~$\alpha,\alpha^*$.  The resulting functions on the
complex $\alpha$\nobreakdash-plane will be real-valued only if $e=0$,
as $\adag^L a\adag^R$ is hermitian only in that case.

Our main result (Theorem~\ref{thm:main}) is the explicit
$s$\nobreakdash-ordering of ${\rm e}^{\lambda(\adag^L a\adag^R)}$,
valid for arbitrary $s\in\mathbb{C}$.  It subsumes eqs.\ (\ref{eq:3})
and~(\ref{eq:4}), which are its $s=\nobreak-1$ and $s=\nobreak+1$
special cases, and includes a new and explicit Weyl ordering formula
($s=\nobreak0$); for the latter, see Theorem~\ref{thm:tired}.  The
general $s$\nobreakdash-ordered representation is expressed in~terms
of what will be called a two-point Hsu--Shiue EGF, interpolating in a
certain sense between the EGF's in (\ref{eq:3}) and~(\ref{eq:4}).

Many examples, including a full treatment of the case $e=1$, are
given.  However, if $e\ge2$ the two-point EGF is obtained by solving a
high-degree polynomial equation (if $e=2$, a~quartic), and may not be
expressible in~terms of radicals.

Any ordering of ${\rm e}^{\lambda(\adag^L a \adag^R)}$ gives an
ordering of $(\adag^L a \adag^R)^n$ by expanding in~$\lambda$.  The
normal or anti-normal ordering of the $n$'th power of a boson string
can be accomplished with the aid of combinatorics~\cite{Mendez2005},
or algorithmically~\cite{Witschel2005}.  In
Proposition~\ref{prop:lastprop} we give a Weyl-ordered representation
of $(\adag a \adag)^n$, derived in a different way that exploits a
known Riordan array.

The paper is structured as follows.  In Sec.~\ref{sec:orderings},
operator orderings including $s$\nobreakdash-ordering are defined,
in part using quantum optics language.  Sheffer polynomial sequences
and exponential Riordan arrays are introduced in
Sec.~\ref{sec:manypolys}, and the parametric Hsu--Shiue family is
defined.  Theorem~\ref{thm:identities34} subsumes eqs.\ (\ref{eq:3})
and~(\ref{eq:4}) above, and serves as a lemma in the derivation by
interpolation of the main result, in Sec.~\ref{sec:main}.  Examples
are given in Sec.~\ref{sec:examples}.

\section{Definitions and $s$-ordering}
\label{sec:orderings}

\subsection{Definitions}

The single-mode Heisenberg--Weyl algebra $\mathcal{W}$ is an
associative algebra over the scalar field~$\mathbb{C}$ generated by
noncommuting indeterminates~$a,\adag$.  Any element
of~$\mathcal{W}$ can be written as a polynomial in~$a,\adag$, i.e.,
a linear combination over~$\mathbb{C}$ of finite products
of~$a,\adag$.  $\mathcal{W}$~has a unit (multiplicative identity),
written~$1$, which is the empty product multiplied by the scalar~$1$.

The representation as polynomials in~$a,\adag$ is not unique, as the
product on~$\mathcal{W}$ satisfies $[a,\adag]=1$.  $\mathcal{W}$ can
be identified with $\mathbb{C}\langle a,\adag\rangle/\allowbreak(a\adag-\nobreak\adag a-\nobreak1)$, the
quotient of the \emph{free} algebra $\mathbb{C}\langle a,\adag\rangle$
by the two-sided ideal generated by $a\adag -\nobreak \nobreak\adag a
- \nobreak1$.  Thus each element of~$\mathcal{W}$ is an equivalence
class of polynomials in~$a,\adag$, though a distinguished
representative can be singled~out by imposing an ordering criterion.

One can also view $\mathcal{W}$ as $\mathbb{C}[\adag][a]$, the ring of
polynomials in~$a$ with coefficients that are polynomials in~$\adag$,
equipped with a product operation that comes from $[a,\adag]=1$.  In
effect, this imposes normal ordering.  Alternatively, one can view
$\mathcal{W}$ as the isomorphic algebra $\mathbb{C}[a][\adag]$, which
imposes anti-normal ordering.  The isomorphism $\mathbb{C}[\adag][a]
\leftrightarrow \mathbb{C}[a][\adag]$ is also based on~$[a,\adag]=1$:
any polynomial can be written in either normal or anti-normal ordered
form.

So, any $G\in\mathcal{W}$ can be viewed as a finite sum
over~$\mathbb{C}$ of terms of the type $\adag^n a^m$, or equivalently
one of terms of the type $a^m \adag^n$.  To avoid the problem of
convergence when dealing with ${\rm e}^G$, which is formally an
infinite sum $\sum_{n=0}^\infty G^n/n!$ and is not an element
of~$\mathcal{W}$, one can embed ${\rm e}^G$ in the one-parameter
family ${\rm e}^{\lambda G}\in\mathcal{W}[[\lambda]]$, a~power series
in an indeterminate~$\lambda$ with coefficients in~$\mathcal{W}$.

It is useful to define an extension algebra
$\overline{\mathcal{W}}{(a)}$ of~$\mathcal{W}$, each element of which
is a sum of $\adag^n a^m$ terms, finite or infinite, but with the
exponent of~$a$ taking only a finite number of values in~all.  (The
product on~$\overline{\mathcal{W}}{(a)}$ comes from $[a,\adag]=1$.)
One can equally well say that each element is a sum of $a^m\adag^n$
terms, finite or infinite, with the same condition imposed.
Similarly, one defines the extension algebra
$\overline{\mathcal{W}}{(\adag)}$ of~$\mathcal{W}$, where it is the
exponent of~$\adag$ that can take only a finite number of values.
$\overline{\mathcal{W}}{(a)}$, $\overline{\mathcal{W}}{(\adag)}$ could
be denoted by $\mathbb{C}[[\adag]][a]$, $\mathbb{C}[[a]][\adag]$.  The
power-series algebras $\mathbb{C}[[\adag]][a][[\lambda]]$,
$\mathbb{C}[[a]][\adag][[\lambda]]$ will also appear.

\subsection{Orderings}

The following defines $s$\nobreakdash-ordering as an interpolation
between normal and anti-normal
ordering~\cite{Cahill69,Schleich2001,Leonhardt97}.  It was introduced
by Cahill--Glauber~\cite{Cahill69} in quantum optics, where $a,\adag$
are ladder (annihilation and creation) operators for a quantized
radiation field. But the treatment here is purely algebraic and does
not single~out a \emph{representation} of~$\mathcal{W}$, such as one
on a Fock space.  So it makes no~mention of coherent states or number
(i.e., Fock) states, and though the Glauber displacement operator
appears as an algebraic object, it will not act on anything.

Let $G\in\mathcal{W}$, i.e., let $G(a,\adag)$ be a polynomial in the
noncommuting $a,\adag$ that represents~$G$ (and is not unique).  One
can define two c\nobreakdash-number or classical representations
(sometimes called transforms) associated to~$G$, denoted by
$F^{(N)},F^{(A)}\in\mathcal{W}_c=\mathbb{C}[\alpha,\alpha^*]$, the
ring of polynomials in \emph{commuting} indeterminates
$\alpha,\alpha^*$, which is
$\mathbb{C}\langle\alpha,\alpha^*\rangle/\allowbreak (\alpha \alpha^*
- \nobreak \alpha^*\alpha)$.  They are defined by requiring that
\begin{subequations}
\begin{align}
  G(a,\adag) &= \normord{F^{(N)}(\alpha,\alpha^*)}_N,\\
  G(a,\adag) &= \normord{F^{(A)}(\alpha,\alpha^*)}_A,
\end{align}
\end{subequations}
as equalities in~$\mathcal{W}$.  In this, the normal and anti-normal
double dot operations $\normord{\dots}_N$ and $\normord{\dots}_A$,
each linear over~$\mathbb{C}$ in its single argument, map
$\mathcal{W}_c$ to~$W$ by quantizing any monomial ${\alpha^*}^n
\alpha^m$ in~$\mathcal{W}_c$ to the elements $\adag^n a^m$ and
$a^m\adag^n$ of~$\mathcal{W}$, respectively.  Thus the inverse images
$F^{(N)},F^{(A)}\in\mathcal{W}_c$ are standardized, ordered
representations of $G\in \mathcal{W}$.

Within any ordering $\normord{\dots}$ one often writes $a,\adag$
instead of $\alpha,\alpha^*$, as was done in the Introduction, there
being no~ambiguity: by convention, within any $\normord{\dots}$ the
indeterminates $a,\adag$ commute, so that writing them as
$\alpha,\alpha^*$ is optional.  Note that the computation of
$F^{(N)},F^{(A)}$ extends from $G\in\mathcal{W}$ to
$G\in\mathcal{W}[[\lambda]]$, in which case
$F^{(N)},F^{(A)}\in\mathcal{W}_c[[\lambda]]$.

In the same way, one defines a standardized Weyl-ordered
representation $F^{(W)}\in\mathcal{W}_c$ of any $G\in\mathcal{W}$ by
requiring that $G(a,\adag)$ equal
$\normord{F^{(W)}(\alpha,\alpha^*)}_W$.  Here, the Weyl quantization
$\normord{\dots}_W$ maps any monomial ${\alpha^*}^n \alpha^m$ to the
\emph{average} of the $\binom{n+m}n = \binom{n+m}m$ distinct orderings
of a product of $n$ $\adag$'s and $m$~$a$'s.  (An equivalent
definition is that in~$\mathcal{W}$,
$\normord{(\mu\alpha+\nobreak\nu\alpha^*)^p}_W$ must equal $(\mu
a+\nobreak \nu \adag)^p$.)  For instance,
\begin{equation}
\normord{{\alpha^*}^2 \alpha}_W = (\adag^2 a + \adag a \adag + a\adag^2)/3,
\end{equation}
where the right-hand side, an element of~$\mathcal{W}$, can be
rewritten with the aid of $[a,\adag]=1$ in various ways, including
$\adag^2 a+\nobreak\adag$, $\adag a \adag$, $a\adag^2-\nobreak \adag$.
So
\begin{equation}
\label{eq:reds}
G(a,\adag)=\adag a\adag \quad\Longrightarrow\quad 
\left\{
\begin{aligned}
F^{(N)}(\alpha,\alpha^*) &= {\alpha^*}^2 \alpha + \alpha^*,\\
F^{(W)}(\alpha,\alpha^*) &= {\alpha^*}^2 \alpha,\\
F^{(A)}(\alpha,\alpha^*) &= {\alpha^*}^2 \alpha -\alpha^*,
\end{aligned}
\right.
\end{equation}
and if $G(a,\adag)=\adag^2 a$, then $\alpha^*$ would be subtracted
from each $F(\alpha,\alpha^*)$, etc.

It is useful to define a formal power series
$D(\beta,\beta^*)\in\mathcal{W}[[\beta,\beta^*]]$ by
\begin{equation}
\begin{aligned}
  D(\beta,\beta^*) &= \exp(\beta\adag-\beta^* a) = \sum_{p=0}^\infty \frac1{p!}(\beta\adag-\beta^* a)^p \\
  &= \sum_{n,m=0}^\infty \frac{\beta^n (-\beta^*)^m}{n!\,m!}\, \normord{{\alpha^*}^n\alpha^m}_W,
\end{aligned}
\end{equation}
so that
\begin{equation}
  \normord{{\alpha^*}^n \alpha^m}_W =
  \frac{\partial^{n+m}}{\partial\beta^n\partial(-\beta^*)^m}
  D(\beta,\beta^*)\Bigm|_{(\beta,\beta^*)=(0,0)}.
\end{equation}
In quantum
optics, $D(\beta,\beta^*)$ appears as the displacement operator.

In such equations as the preceding, a notational convention common in
quantum optics is adopted: paired quantities such as $\alpha,\alpha^*$
and $\beta,\beta^*$, which in some contexts are treated as a complex
number and its conjugate, are treated in others as being independent.
(Not least, in the exponentiation of~$\mathcal{W}$ to the complexified
Heisenberg--Weyl group~\cite{Wunsche91}.)  For greater rigor, one
could instead introduce commuting quantities $\alpha_r,\alpha_i$ and
$\beta_r,\beta_i$, the real and imaginary parts of~$\alpha$
and~$\beta$, to serve as real coordinates on the complex $\alpha$- and
$\beta$\nobreakdash-planes~\cite{Schleich2001}.

Define an $s$\nobreakdash-parametrized deformation of the exponential
operator $D(\beta,\beta^*)$, namely $D(\beta,\beta^*;s)= {\rm
  e}^{s\beta(-\beta^*)/2}D(\beta,\beta^*)$, where $s\in\mathbb{C}$;
and for each $n,m$, define the $s$\nobreakdash-quantization
$\normord{{\alpha^*}^n\alpha^m}_s$ of ${\alpha^*}^n \alpha^m$ by
\begin{equation}
  D(\beta,\beta^*;s) = \sum_{n,m=0}^\infty \frac{\beta^n (-{\beta^*})^m}{n!\,m!}\, 
\normord{{\alpha^*}^n\alpha^m}_s,
\end{equation}
so that
\begin{equation}
\label{eq:findsthat}
  \normord{{\alpha^*}^n \alpha^m}_s = \frac{\partial^{n+m}}{\partial\beta^n\partial(-\beta^*)^m} D(\beta,\beta^*;s)\Bigm|_{(\beta,\beta^*)=(0,0)}.
\end{equation}
Equivalently, one requires that
\begin{equation}
  D(\beta,\beta^*) = \exp(\beta\adag-\beta^* a) = {\rm
    e}^{s|\beta|^2/2}\normord{{\rm e}^{\beta \alpha^*-\beta^* \alpha}}_s,
\end{equation}
as an equality in~$\mathcal{W}[[\beta,\beta^*]]$.

As an example of~(\ref{eq:findsthat}), one finds with the aid of
$[a,\adag]=1$ that 
\begin{equation}
\normord{{\alpha^*}^2 \alpha}_s =\adag a\adag + s\,\adag.
\end{equation}
The $s$\nobreakdash-quantization operator $\normord{\dots}_s$ extends
by linearity from monomials in the (commuting)
indeterminates~$\alpha,\alpha^*$, to polynomials in same, i.e., to
elements of ${\mathcal{W}}_c=\mathbb{C}[\alpha,\alpha^*]$, and to
elements of $\mathcal{W}_c[[\lambda]] =
\mathbb{C}[\alpha,\alpha^*][[\lambda]]$, as well.

Thus, one can define a unique $s$\nobreakdash-ordered representation
or $s$\nobreakdash-transform of any $G=G(a,\adag)$ in~$\mathcal{W}$,
denoted by $F^{(s)}(\alpha,\alpha^*)$, by requiring that
\begin{gather}
\label{eq:predual}
  G(a,\adag) = \normord{F^{(s)}(\alpha,\alpha^*)}_s.
  \\
  \shortintertext{So, for example,}
  \label{eq:foo0}
  G(a,\adag)=\adag a\adag\quad\Longrightarrow\quad
  F^{(s)}(\alpha,\alpha^*) = {\alpha^*}^2 \alpha - s\,\alpha^* .
\end{gather}
When $s=\nobreak0$, $F^{(s)}$ reduces to the Weyl-ordered representation
$F^{(W)}$.  It is also not difficult to show, using the
Baker--Campbell--Hausdorff formula
\begin{equation}
  {\rm e}^{A}   {\rm e}^{B} =   {\rm e}^{A+B+[A,B]/2},
\end{equation}
which holds as an equality between power series if $[A,B]$ commutes
with $A$ and~$B$, that when $s=\nobreak-1$ and $s=\nobreak+1$,
$F^{(s)}$ reduces respectively to the normal and anti-normal ordered
representations, $F^{(N)}$ and~$F^{(A)}$.  The $s=-1,0,+1$ cases
of~(\ref{eq:foo0}) are reflected in eq.~(\ref{eq:reds}).

In modern quantum optics, the most commonly encountered
c\nobreakdash-number representations of any element of the
Heisenberg--Weyl algebra are the classical functions
$F^{(-1)}=F^{(N)}$, $F^{(0)}=F^{(W)}$, $F^{(+1)}=F^{(A)}$
on~$\mathbb{C}$ that arise from $G=\hat\rho$, where $\hat\rho$~is a
specified (single-mode) density operator that is an element
of~$\mathcal{W}$, or an extension of~$\mathcal{W}$ that includes
formal power series in~$a,\adag$ \cite{Schleich2001,Lee95,Dariano97}.
These are called the $Q$, $W$, and~$P$ quasiprobability densities, and
are attributed respectively to Husimi, Wigner, and Glauber--Sudarshan.
But in older works on phase-space methods in quantum optics, operators
other than density operators and observables were frequently
``classicized,'' and it was in that general context that
$s$\nobreakdash-ordering was introduced.  (See
\cite{Mehta68,Agarwal70,Mehta77}, and more
recently~\cite{Fan2011}.\footnote{There are two conventions on the
sign of~$s$.  In the original paper of Cahill--Glauber~\cite{Cahill69}
and a few more recent ones \cite{Fan2011,Fan2010,Shahandeh2012},
$s=+1,-1$ correspond to normal and anti-normal ordering; e.g., for a
density operator, to $Q$ and~$P$.  But usually, they correspond to
anti-normal and normal ordering (e.g., to $P$ and~$Q$).  See, e.g.,
\cite[Chapter~12]{Schleich2001}, \cite{Orlowski93,MoyaCessa93}, and
\cite[Chapter~3]{Leonhardt97}.  The latter convention has been adopted
here, so the sign of~$s$ in eq.~(\ref{eq:CH2}), taken
from~\cite{Cahill69}, has been changed.  Note also that in some works
(e.g., \cite{Lee95} and \cite[Table~V]{Agarwal70}), the meanings of
the so-called reciprocal orderings $F^{(N)}$ and~$F^{(A)}$, related by
negation of~$s$, are interchanged.})

Differentiating the exponential operator $D(\beta,\beta^*;s)={\rm
  e}^{(s'-s)|\beta|^2/2} D(\beta,\beta^*;s')$ with respect to
$\beta,\beta^*$, in a procedure that can be traced to
Wilcox~\cite{Wilcox67}, and using~(\ref{eq:findsthat}), yields the
$s$\nobreakdash-ordering conversion formula
\begin{equation}
\label{eq:genconv}
  \normord{{\alpha^*}^n\alpha^m}_s = \sum_{k=0}^{\min(m,n)} k!\,\binom{n}k\binom{m}k
  \left(\frac{s-s'}2\right)^k
  \normord{{\alpha^*}^{n-k}\alpha^{m-k}}_{s'}.
\end{equation}
The $n=m=1$ case of~(\ref{eq:genconv}) is the contraction
\begin{equation}
  \normord{\alpha^*\alpha}_s -   \normord{\alpha^*\alpha}_{s'}  = \frac{s-s'}2,
\end{equation}
from which the general formula (\ref{eq:genconv}) can be derived, much
as Wick's theorem is derived~\cite{Shahandeh2012}.  As another
example, the $s=\nobreak+1$, $s'=-1$ case of~(\ref{eq:genconv}) is
\begin{equation}
a^m  \adag^n = \sum_{k=0}^{\min(m,n)} k!\,\binom{n}k\binom{m}k
    \adag^{n-k}a^{m-k},
\end{equation}
an identity in $\mathcal{W}$ that converts an anti-normal ordering
($s=\nobreak+1$) to a normal one ($s=\nobreak-1$).  Thus, one can
convert between the coefficients of an anti-normal and a normal
series, $\sum_{m,n} f_{m,n}^{(A)} a^m \adag^n$ and $\sum_{m,n}
f_{n,m}^{(N)} \adag^n a^m$ \cite{Shalitin79}.

Equation~(\ref{eq:genconv}) can be written in~terms of a two-variable
Hermite polynomial~\cite{Fan2010,Shahandeh2012}, as was recognized
early (when $s\in\{0,\pm1\}$) by Agarwal and Wolf
\cite[Table~V]{Agarwal70}.  It can also be written as
\begin{equation}
\label{eq:predual2}
\begin{aligned}
  \normord{{\alpha^*}^n \alpha^m}_s 
  &= \sum_{k=0}^{\min(m,n)} \frac1{k!}
  \left(\frac{s-s'}2\right)^k
  \normord{ \left( \frac{\partial ^2}{\partial \alpha \partial \alpha^* }   \right)^k
    {{\alpha^*}^n \alpha^m}}_{s'} \\
    &=
    \normord{ \exp\left[
        \left(\frac{s-s'}2
        \right)
        \frac{\partial^2}{\partial \alpha \partial \alpha^*}
        \right]
        {{\alpha^*}^n \alpha^m}}_{s'}.
\end{aligned}
\end{equation}
This extends by linearity to
\begin{equation}
  \normord{F}_s =
  \normord{ \exp\left[
        \left(\frac{s-s'}2
        \right)
        \frac{\partial^2}{\partial \alpha \partial \alpha^*}
        \right]
        F}_{s'}
\end{equation}
for any $F\in\mathcal{W}_c$, i.e., for any polynomial $F=F(\alpha,\alpha^*)$.

Suppose now that
$F^{(s)},F^{(s')}\in{\mathcal{W}}_c=\mathbb{C}[\alpha,\alpha^*]$ are
the $s$-ordered and $s'$-ordered ``c\nobreakdash-number'' versions of
some $G\in\mathcal{W}$, computed from~$G$ by the
requirement~(\ref{eq:predual}).  By~(\ref{eq:predual2}), they must be
related to each other by the dual relation
\begin{equation}
\label{eq:finitesum}
  F^{(s)}(\alpha,\alpha^*) = 
\exp\left[
        \left(\frac{s'-s}2
        \right)
        \frac{\partial^2}{\partial \alpha \partial \alpha^*}
        \right]
F^{(s')}(\alpha,\alpha^*),
\end{equation}
in which the exponentiated second-order differential operator has been
inverted.  In the formally infinite series in~(\ref{eq:finitesum})
obtained by expanding the exponential in a power series, only a finite
number of terms will be nonzero.  Equivalently, $F=F^{(s)}(\alpha,\alpha^*)$
satisfies the diffusion or heat equation
\begin{equation}
  \label{eq:unexponentiated}
  \frac{\partial F}{\partial (-s)} = \frac12\, \frac{\partial^2}{\partial\alpha\partial\alpha^*}F,
\end{equation}
an unexponentiated version of~(\ref{eq:finitesum}).  (See
\cite[Chapter~12]{Schleich2001}.)

Equation~(\ref{eq:finitesum}) can be extended to the case when
$F^{(s)},F^{(s')}$ are convergent power series in~$\alpha,\alpha^*$
rather than polynomials, at~least formally.  The nicest case is when
each of $F^{(s)},F^{(s')}$, regarded as a Maclaurin expansion at
$(\alpha,\alpha^*)=(0,0)$, i.e., $(\alpha_r,\alpha_i)=(0,0)$,
converges globally and defines an analytic function of the real
variables~$\alpha_r,\alpha_i$, or even a complex analytic function of
$\alpha= \alpha_r+\nobreak{\rm i}\alpha_i$.  Because
$\partial^2/\partial\alpha\partial\alpha^*$ equals $\Delta/4$, where
$\Delta=\partial^2/\partial\alpha_r^2 + \partial^2/\partial\alpha_i^2$
is the Laplacian on the complex $\alpha$\nobreakdash-plane coordinatized
by~$\alpha_r,\alpha_i$, one has
\begin{equation}
\begin{aligned}
\label{eq:smoothing}
  F^{(s)}(\alpha,\alpha^*) &= {\rm e}^{(s'-s)\Delta/8}F^{(s')}(\alpha,\alpha^*)
  \\
  &\propto \int F^{(s')}(\beta,\beta^*)\, {\rm e}^{-2|\beta-\alpha|^2/(s'-s)}\,{\rm d}^2\beta.
\end{aligned}
\end{equation}
So in this case, when $s<s'$ it may be possible to compute $F^{(s)}$
from $F^{(s')}$ by a smoothing \emph{Weierstrass transform} operation
on the complex $\alpha$\nobreakdash-plane: convolution with a Gaussian
kernel of variance $(s'-\nobreak s)/4$.

This operation is familiar from quantum
optics~\cite{Schleich2001,Leonhardt97,Sperling2020,Wunsche91}: among
the quasiprobability densities~$F^{(s)}$ obtained when $G=G(a,\adag)$
is a density operator~$\hat\rho$, the Wigner $W$~function
($s=\nobreak0$) is a smoothed version of the Glauber--Sudarshan
$P$~function ($s=\nobreak+1$), and the Husimi $Q$~function
($s=\nobreak-1$) is a smoothed version of the $W$~function.  The
smoothing of~$P$ to $W$ and~$Q$ has been extended from functions of
$\alpha\in\mathbb{C}$ to tempered and other distributions on the
complex $\alpha$\nobreakdash-plane, and the variation with~$s$ of the
mathematical properties and physical interpretation of smoothed
quasiprobability densities has been
explored~\cite{Schleich2001,Orlowski93,Wunsche91}.

In the proof of the main result of this paper
(Theorem~\ref{thm:main}), the unexponentiated
form~(\ref{eq:unexponentiated}), and not the exponentiated forms
(\ref{eq:finitesum}),(\ref{eq:smoothing}), will be employed to relate
$F^{(s)}$, $F^{(s')}\in \mathcal{W}_c[[\lambda]]$, to obviate 
consideration of series convergence.

\section{Riordan, Sheffer, and Hsu--Shiue polynomials}
\label{sec:manypolys}

The following definitions and facts are standard.  For the theory of
(exponential) Riordan arrays and polynomials, introduced under that
name by Shapiro et~al.\ to derive combinatorial identities, see
\cite{Shapiro91,Barry2016,Shapiro2022}.  For the equivalent theory of
Sheffer polynomial sequences (originally called Sheffer of $A$-type
zero), which are used in approximation theory and elsewhere, see
\cite{Roman84,Costabile2022}.  In their present form, both emerged
from the work of G.-C. Rota and collaborators on the refashioning of
the classical umbral calculus as a finite operator
calculus~\cite{Rota75,DiBucchianico95}.

Hsu and Shiue~\cite{Hsu98} defined a parametric family of generalized
Stirling numbers, which now have their own literature
\cite[Chapter~4]{Mansour2016}.  For any choice of parameters, they can
be viewed as the elements of a exponential Riordan array, the ``row
polynomials'' of which are Sheffer polynomials that generalize the
classical Touchard polynomials, which have Stirling numbers as
coefficients.

\subsection{Exponential Riordan arrays and polynomials}

Let $d,h\in\mathbb{C}[[z]]$ be formal power series in~$z$ of order $0$
and~$1$ respectively, meaning that their first nonzero terms are the
$z^0$ and $z^1$ terms; so $1/d$ and the compositional inverse~$\bar h$
both exist as formal power series.  (For instance, $d,h$ may be
Maclaurin expansions of functions $d(z),h(z)$ analytic at~$z=0$.)  One
usually requires that $d$~be monic, $d=1\cdot z^0+\nobreak O(z^1)$,
and often that $h=1\cdot z^1+\nobreak O(z^2)$.  Then, a polynomial
sequence $s_n(t)$, $n\ge0$, where $\deg s_n=n$, with EGF
$\mathcal{S}(t,z)= d(z){\rm e}^{th(z)}$ that is an element of
$\mathbb{C}[t][[z]]$, is defined by
\begin{equation}
  \sum_{n=0}^\infty s_n(t)\frac{z^n}{n!} = \mathcal{S}(t,z),
\end{equation}
with $s_0(t)\equiv1$.  The sequence $s_n(t)$ is said to be exponential
Riordan for the ordered pair~$[d,h]$, and will be denoted briefly
by~$\mathbf{R}[d,h]$.  Such notations as $s_n(t)=\mathbf{R}[d,h]_n(t)$
and $\mathcal{S}(t,z)=\mathbf{R}[d,h](t,z)$ are also used.

Equivalently, one can define a triangular array of
coefficients $s_{n,k}$, $0\le k\le n$, by
\begin{equation}
  \sum_{n=0}^\infty \sum_{k=0}^n s_{n,k}t^k\frac{z^n}{n!} = \mathcal{S}(t,z),
\end{equation}
so that
\begin{equation}
\label{eq:expRdef}
  s_{n,k} = \frac{n!}{k!}\, [z^n] \left[d(z)h(z)^k\right],
\end{equation}
where $[z^n]$ extracts the coefficient of~$z^n$ from the element of
$\mathbb{C}[[z]]$ that follows.  Then $s_n(t)=\sum_{k=0}^n
s_{n,k}t^k$, so the $s_n(t)$ are the row polynomials of the array.
The array $s_{n,k}$ can be viewed as the lower-triangular part of an
infinite matrix, and is called an exponential Riordan array.  The
notation $s_{n,k}=\mathbf{R}[d,h]_{n,k}$ will also be used, with
$\mathbf{R}[d,h]$ alternatively signifying the matrix as a whole.

Row-finite lower-triangular matrices form a group under matrix
multiplication, and exponential Riordan arrays turn~out to form a
subgroup, which will be denoted by~$\mathfrak{R}$.  The following is a
fundamental fact.

\begin{proposition}[\cite{Shapiro2022}]
In the exponential Riordan group\/ $\mathfrak{R}$ the product operation
is\/ $\mathbf{R}[d_1,h_1]\,\mathbf{R}[d_2,h_2] = \allowbreak
\mathbf{R}[(d_2\circ\nobreak h_1)d_1, \allowbreak h_2\circ\nobreak
  h_1]$, and the inversion operation is\/ $\mathbf{R}[d,h]^{-1}
=\allowbreak \mathbf{R}[1/(d\circ\nobreak \bar h), \bar h]$, where
$\bar h$~is the compositional inverse of~$h$, i.e., $\bar
h(h(z))=\nobreak z$.  The multiplicative identity is\/ $[1,z]$, which
signifies the infinite identity matrix.
\end{proposition}

For matrix inverses, the notation
\begin{equation}
\mathbf{R}[d,h]^{-1} \eqdef \mathbf{R}\overline{[d,h]} = \mathbf{R}[\bar d,\bar h] ,
\end{equation}
where $\bar d=1/(d\circ \bar h)$, can be used.  Note that $\bar d,\bar
h$, like~$d,h$, are formal power series of orders $0$ and~$1$.

A quantum-optics example of inversion in~$\mathfrak{R}$ is afforded by
the Cahill--Glauber identity (\ref{eq:CH2}) and its
inverse~(\ref{eq:invCH2}): expansion in the parameter~$\lambda$
reveals that they are based respectively on the arrays
\begin{sizeequation}{\small}
  \mathbf{R}\left[
\frac2{1+{\rm e}^z - s(1-{\rm e}^z)}, \frac{2({\rm e}^z-1)}{1+{\rm e}^z - s(1-{\rm e}^z)}\right], 
  \quad
  \mathbf{R}
  \left[
    \frac{2}{2-z-sz}, \ln\left(\frac{2+z-sz}{2-z-sz}\right)
    \right]
\end{sizeequation}
in~$\mathfrak{R}$, which are inverses of each other.

The generating functions of the Riordan-array theory facilitate the
computation of matrix--vector products, as follows.

\begin{proposition}[\cite{Shapiro2022}]
\label{prop:mv}
For any column vector $u=(u_k)$ with
EGF $u(t)=\sum_{k=0}^\infty u_kt^k/k!$, an element
of\/ $\mathbb{C}[[t]]$, and for any exponential Riordan array $s_{n,k}$
equaling\/ $\mathbf{R}[d,h]_{n,k}$, define the column vector $v=(v_n)$ by
\begin{equation}
  v_n = \sum_{k=0}^n s_{n,k} u_k,
\end{equation}
or symbolically by\/ $v=\mathbf{R}[d,h]u$.  Then, its EGF
$v(z)=\sum_{n=0}^\infty v_nz^n/n!$, an element of~\/$\mathbb{C}[[z]]$,
equals $d(z) u(h(z)) = [d\cdot (u\circ h)](z)$, where\/
$\cdot$~and\/~$\circ$ are the binary operations of product and
composition on formal power series.
\end{proposition}

\subsection{Sheffer polynomials}
\label{sec:Sheffer}

The theories of Sheffer and exponential Riordan polynomials are
equivalent, being related by the inversion operation
\cite{He2007,Wang2008}.  A~polynomial sequence $s_n(t)$, $n\ge0$,
where $\deg s_n=n$ and usually $s_0(t)\equiv1$, is said to be Sheffer
for the ordered pair $[g,f]$ if and only if it is exponential Riordan
for the pair $[d,h]$ equal to $\overline{[g,f]}$, the inverse of
$[g,f]$ in the group~$\mathfrak{R}$, which will be written as~$[\bar
  g,\bar f]$.  The sequence as a whole will be denoted briefly by
$\mathbf{S}[g,f]$, or equivalently by $\mathbf{R}[d,h] =
\mathbf{R}[\bar g,\bar f]$.

It is convenient in the Sheffer theory to use the indeterminate~$D$ in
the series $g,f$, and $z$ in the series~$\bar g,\bar f$.  That is, to
reduce ambiguity one adopts the convention that
$g,f\in\mathbb{C}[[D]]$ and $\bar g,\bar f\in\mathbb{C}[[z]]$.  This
makes it possible to write both $z=f(D)$ and $D=\bar f (z)$, related
by series reversion.

So a polynomial sequence~$s_n$ is Sheffer for~$[g,f]$, where
$g,f\in\mathbb{C}[[D]]$ are of order $0$ and~$1$ respectively, and
$\bar g,\bar f\in\mathbb{C}[[z]]$ are likewise, if
\begin{equation}
\label{eq:appearingin}
  \sum_{n=0}^\infty s_n(t)\frac{z^n}{n!} = \bar g(z){\rm e}^{t\bar
    f(z)}= \frac{1}{g(\bar f(z))} {\rm e}^{t\bar f(z)},
\end{equation}
and one writes $s_n=\mathbf{S}[g,f]_n=\mathbf{R}[\bar g,\bar f]_n$ and
$s_{n,k}=\mathbf{S}[g,f]_{n,k} = \mathbf{R}[\bar g,\bar f]_{n,k}$, etc.

Many well-known polynomial sequences can be viewed as Sheffer
sequences \cite{Roman84,Kim2013a}, and there is some overlap with the
orthogonal polynomials of theoretical physics.  An example is the
sequence of Hermite polynomials $\Hermiteprob_n(t)$, which is
$\mathbf{S}[{\rm e}^{D^2/2},D]$ or equivalently $\mathbf{R}[{\rm
    e}^{-z^2/2},z]$.  The $\Hermiteprob_n(t)$ are orthogonal with
respect to an appropriate measure (namely, ${\rm e}^{-t^2/2}\,{\rm
  d}t$), as they happen to satisfy a three-term recurrence relation.
The squared Hermite polynomials $[\Hermiteprob_n(\sqrt{t})]^2$
turn~out to be $\mathbf{S}[\sqrt{1-\nobreak 2D}/\allowbreak(1-\nobreak
  D), \allowbreak D/(1-\nobreak D)]$ or $\mathbf{R}[1/\sqrt{1-\nobreak
    z^2}, \allowbreak z/\allowbreak(1+\nobreak z)]$, but do not
satisfy such a relation.

For a Sheffer sequence $\mathbf{S}[g,f]$, the formal power series
$z=f(D)$ may be the Maclaurin expansion of an elementary function
of~$D$, without the compositional inverse $D=\bar f(z)$ appearing in
the EGF~(\ref{eq:appearingin}) having the same property.  An example
is the sequence of Abel polynomials $A_n(t)=t(t-\nobreak n)^{n-1}$,
which turns~out to be ${\bf S}[1,D{\rm e}^D]$.  (See~\cite{Roman84}.)
But the compositional inverse of $z=f(D)=D{\rm e}^D$ is the Lambert
$W$~function $D=\bar f(z)=\allowbreak W_L(z)=\allowbreak z-\nobreak
z^2 +\nobreak \frac32z^3+\nobreak \cdots$, which is not an elementary
function. It is more convenient to refer to the Abel sequence as ${\bf
  S}[1,D{\rm e}^D]$ than as ${\bf R}[1,W_L(z)]$.  The same phenomenon
will appear in Sec.~\ref{sec:main}.

The following reveals the deep connection between the Sheffer theory
and the Heisenberg--Weyl algebra.  It also clarifies why the seemingly
arbitrary symbol~`$D$' is used here as a power-series variable.

\begin{proposition}[\cite{Roman84,Blasiak2006,VerdeStar2025}]
  For any pair $g,f\in\mathbb{C}[[D]]$ of orders $0,1$ that defines a
  Sheffer sequence $s_n=\mathbf{S}[g,f]_n$ of polynomials in~$t$, the
  Heisenberg--Weyl algebra generated by $P, M$ with\/ $[P,M]=1$ has a
  representation on\/~$\mathbb{C}[t]$, the space of polynomials in~$t$,
  given by lowering and raising operators
  \begin{displaymath}
    P = f(D),\qquad  M = \left[t-\frac{g'(D)}{g(D)}\right]\frac1{f'(D)},
  \end{displaymath}
  the action of which on the Sheffer polynomials is
  \begin{equation}
  \label{eq:useasalt}
    Ps_n=ns_{n-1},\qquad M s_n=s_{n+1}.
  \end{equation}
  Here, $D$ is interpreted as $D_t = {\rm d}/{\rm d}t$, and though
  $P,M$ may be of infinite order in~$D$, as differential operators
  they map\/ $\mathbb{C}[t]$ into itself.
\end{proposition}

The operators $P,M$ are called monomiality operators, and under a
suitable invertible linear transformation of~$\mathbb{C}[t]$, the
polynomials~$s_n$ can be taken to the monomials~$t^n$, and $P,M$
to~$D_t,t$.  (See~\cite{VerdeStar2025} and references therein.)
Another differential statement involving $D=D_t$ is
\begin{displaymath}
  [t^0]\,g(D_t){\rm e}^{\lambda f(D_t)} s_n(t) = \lambda^n,
\end{displaymath}
an equality in $\mathbb{C}[t][[\lambda]]$ which
like~(\ref{eq:useasalt}) can be used as an alternative definition of
the polynomials~$s_n$~\cite{Roman84}.

The exponentiated raising operator ${\rm e}^{\lambda M}$ is of
interest.  Each non-negative power~$M^n$ is an element of
$\mathbb{C}[t][[D_t]]$, and because
$M^ns_0=s_n$ and $s_0(t)\equiv1$, it follows
from~(\ref{eq:appearingin}) that
\begin{equation}
\left[{\rm e}^{\lambda M}1\right](t) = \bar g(\lambda){\rm e}^{t\bar f(\lambda)},  
\end{equation}
an equality in~$\mathbb{C}[t][[\lambda]]$.  The right-hand side is the
Sheffer EGF at $z=\lambda$.

For later reference, note that replacing $D,t$ by $a,\adag$
converts~$M$ to an element~$\tilde M$ of the algebra
$\mathbb{C}[[a]][\adag]$, which is strictly larger than
$\mathbb{C}[a][\adag]$, the realization of~$\mathcal{W}$ by
anti-normal ordered polynomials in~$a,\adag$.

\subsection{Hsu--Shiue arrays and polynomials}
\label{subsec:HS}

Hsu--Shiue arrays $\mathcal{HS}(A,B;r)$, where $A,B,r\in\mathbb{C}$,
are a parametric family of exponential Riordan arrays, i.e., a family
of elements of the group~$\mathfrak{R}$, with associated row
polynomials and~EGF's.  They can be introduced as follows.

The classical Stirling numbers (of the second kind), denoted by
$\left\{{n \atop k}\right\}$ or $S(n,k)$, count the partitions of an
$n$\nobreakdash-set into $k$~disjoint, nonempty subsets.  They make~up
a triangular array indexed by $n,k$ with $0\le k\le n$, and satisfy
\begin{equation}
  x^n = \sum_{k=0}^n \left\{{n \atop k}\right\} x^{\underline k},
\end{equation}
where $x^{\underline k} = x(x-1)\dots(x-k+1)$ is the falling factorial
with unit ``stride.''  So, they are coefficients of connection between
bases of~$\mathbb{C}[x]$.  The sequence of row polynomials
$\sum_{k=0}^n \left\{{n \atop k}\right\} t^k$, $n\ge0$, which could be
called Stirling polynomials (of~the second kind), but are more often
called Touchard, Bell, or exponential polynomials, is the
exponential Riordan sequence ${\bf R}[1,\allowbreak {\rm
    e}^z-\nobreak1]$ or equivalently the Sheffer sequence ${\bf
  S}[1,\allowbreak \ln(1+\nobreak D)]$.

The generalized Stirling numbers of Hsu--Shiue~\cite{Hsu98} are a
family $\HS_{n,k}=\HS_{n,k}(A,B,r)$ that subsumes the classical
Stirling numbers and many previous generalizations.  For
$A,B,r\in\mathbb{C}$, they are defined by
\begin{equation}
\label{eq:subsumes}
(x+r)^{\underline{n},A} = \sum_{k=0}^n HS_{n,k}(A,B,r)\, x^{\underline{k},B},
\end{equation}
where $(x+r)^{\underline{n},A}$, $x^{\underline{k},B}$ are falling
factorials with stride not necessarily equal to unity, e.g.,
$x^{\underline{k},B} = x(x-\nobreak B)\dots[x-\nobreak (k-\nobreak
  1)B]$.  The $\HS_{n,k}$ reduce to $\left\{{n \atop k}\right\}$ if
$(A,B,r)=(0,1,0)$.  The parametric row polynomials
\begin{equation}
\HS_{n}(t;\,A,B,r) = \sum_{k=0}^n \HS_{n,k}(A,B,r) t^k, \qquad n\ge0,
\end{equation}
will be called Hsu--Shiue polynomials.  (Note that $\HS_0(t)\equiv1$
and $\deg\HS_n=n$ for all~$n$.)  

By examination, this $n$-indexed sequence of polynomials is
exponential Riordan for the pair
$[d_{\textit{HS}}(\cdot;A,r),h_{\textit{HS}}(\cdot;A,B)]$, where if
$A,B\neq0$,
\begin{equation}
\label{eq:dhformulae}
  \begin{aligned}
    &d_{\textit{HS}}(z;A,r) \defeq (1+Az)^{r/A},
    \\ 
    &h_{\textit{HS}}(z;A,B)  \defeq \frac{(1+Az)^{B/A}-1}B,
  \end{aligned}
\end{equation}
the right-hand sides being elements of $\mathbb{C}[[z]]$.  (The cases
$A=0$, $B=0$ are handled by taking limits.)  One way of confirming
this is to verify that eq.~(\ref{eq:subsumes}), viewed as an equality
between a column vector indexed by~$n$ and a matrix--vector product,
is consistent with Proposition~\ref{prop:mv}.

It is easily checked that
\begin{equation}
  \overline{[d_{\textit{HS}}(\cdot;A,r),h_{\textit{HS}}(\cdot;A,B)]}
={[d_{\textit{HS}}(\cdot;B,-r),h_{\textit{HS}}(\cdot;B,A)]},
\end{equation}
which is a statement of \emph{duality}: as infinite lower-triangular
matrices and as elements of the exponential Riordan
group~$\mathfrak{R}$, Hsu--Shiue arrays $\mathcal{HS}(A,B,r)$ and
$\mathcal{HS}(B,A,-r)$ are inverses of each other.
(Cf.~\cite{Hsu98}.)  So, one can write
\begin{align}
&
\begin{aligned}
\mathcal{HS}(A,B,r) &= {\bf R}[d_{\textit{HS}}(z;A,r),h_{\textit{HS}}(z;A,B)]\\
&= {\bf S}[d_{\textit{HS}}(D;B,-r),h_{\textit{HS}}(D;B,A)]
\end{aligned}
\\
\shortintertext{and}
&
\begin{aligned}
\HS_{n}(t;\,A,B,r) &= {\bf R}[d_{\textit{HS}}(z;A,r),h_{\textit{HS}}(z;A,B)]_n(t)\\
&= {\bf S}[d_{\textit{HS}}(D;B,-r),h_{\textit{HS}}(D;B,A)]_n(t).
\end{aligned}
\label{eq:usedforL}
\end{align}
Regardless of whether one uses the exponential Riordan or the Sheffer
notation to express the~$\HS_n$, the following holds.
\begin{proposition}
\label{prop:HSEGF}
The parametric Hsu--Shiue polynomials
$\HS_n(t)$ have the EGF
\begin{displaymath}
\sum_{n=0}^\infty \HS_n(t;\,A,B,r) \frac{z^n}{n!} =
(1+Az)^{r/A} \exp\left\{\frac{t}{B}\left[(1+Az)^{B/A} - 1\right]\right\},
\end{displaymath}
with the cases $A=0$, $B=0$ handled by taking limits.  The EGF can be
viewed as an element of\/ $\mathbb{C}[t][[z]]$, or more strongly as an
analytic function on a neighborhood of\/ $(t,z)=(0,0)$.
\end{proposition}
From this EGF, which will be denoted by $\mathcal{HS}(t,z) =
\mathcal{HS}(t,z;A,B,r)$, the $\HS_n(t)$ can be derived by repeated
differentiation.  For each~$k$ with $0\le k\le n$, the coefficient
$\HS_{n,k}$ of the $k$'th power of~$t$ in~$\HS_n(t)$ will be a
polynomial in $A,B,r$ with integer coefficients.

If $B\neq0$, the $\HS_{n,k}$ can
also be computed from
\begin{equation}
\label{eq:HSsummation}
  \HS_{n,k}(A,B,r) = 
\frac{1}{B^kk!}\sum_{j=0}^k (-1)^{k-j} \binom{k}{j}(Bj+r)^{\underline{n},A},
\end{equation}
which comes from eq.~(\ref{eq:subsumes}) by finite-differencing to
extract coefficients~\cite{Corcino2001}.  And for any $A,B,r$, the
$\HS_{n,k}$ can be computed from the recurrence~(\ref{eq:HSrec}),
below.
\begin{proposition}
  The following are equivalent characterizations of the EGF
  \begin{displaymath}
    \mathcal{HS}(t,z;\,A,B,r) = \sum_{n=0}^\infty \HS_n(t;\,A,B,r)\frac{z^n}{n!}
    = \sum_{n=0}^\infty \sum_{k=0}^n \HS_{n,k}(A,B,r)\,t^k\frac{z^n}{n!},
  \end{displaymath}
which is an element of\/ $\mathbb{C}[t][[z]]$.
\begin{enumerate}
\item[(i)] The explicit formula for\/ $\mathcal{HS}(t,z;A,B,r)$ given
  in Proposition~\ref{prop:HSEGF}.
\item[(ii)] The first-order partial differential equation (PDE)
  \begin{equation}
\left[    (-Az-1)\frac{\partial}{\partial z} + Bt\,\frac{\partial}{\partial t} + (r+t)\right] \mathcal{HS}=0,
  \end{equation}
  on a neighborhood of\/ $(t,z)=(0,0)$, with the initial condition\/
  $\mathcal{HS}(t,0)\equiv1$.
\item[(iii)]
  The triangular recurrence
  \begin{equation}
    \label{eq:HSrec}
    \HS_{n+1,k+1} = \left[-An+B(k+1)+r\right]\HS_{n,k+1} + \HS_{n,k},
  \end{equation}
  with the initial condition that the apex element $\HS_{0,0}$ equal
  unity, and the convention that $\HS_{n,k}=0$ if\/ $k<0$ or\/~$k>n$.
\end{enumerate}
\end{proposition}
\begin{proof}
  That $\mathcal{HS}(t,z)$ as supplied by Proposition~\ref{prop:HSEGF}
  satisfies the PDE follows by direct computation.  The PDE (with
  initial condition) can be solved on a neighborhood of $(t,z)=(0,0)$
  by the Lagrange--Charpit method of characteristics, yielding the
  same analytic formula for $\mathcal{HS}(t,z)$.  And, it is easily
  checked that $\mathcal{HS}(t,z)$, a~function analytic at~$(0,0)$
  with bivariate Maclaurin expansion in $\mathbb{C}[t][[z]]$,
  satisfies the PDE if and only if its expansion coefficients
  $\HS_{n,k}$ satisfy the recurrence~(\ref{eq:HSrec}).
\end{proof}

For some choices of $A,B,r$, explicit expressions for the polynomial
coefficients $\HS_{n,k}(A,B,r)$ and the polynomials $\HS_n(t; A,B,r)$
can be obtained from the summation formula~(\ref{eq:HSsummation}) or
the recurrence~(\ref{eq:HSrec}).  For instance,
\begin{subequations}
  \begin{align}
    & \HS_{n,k}(B,B,r) = \binom{n}{k} r^{\underline{n-k},B},\\
    & \HS_{n,k}(-B,B,r) = \binom{n}{k} (Bk+r)^{\overline{n-k},B},
    \label{eq:43b}
  \end{align}
\end{subequations}
with an obvious notation for the rising factorial.  Notice that
by~(\ref{eq:43b}),
\begin{alignat}{2}
  &\HS_{n,k}(-1,1,1) = {\binom{n}{k}}\frac{n!}{k!}, &\qquad&\HS_{n,k}(1,-1,-1) = (-1)^{n-k}{\binom{n}{k}}\frac{n!}{k!},
\\
\shortintertext{and}
  &\HS_n(t;\,-1,1,1) = n!\,L_n(-t),&\qquad&\HS_n(t;\,1,-1,-1) = (-1)^nn!\,L_n(t),
\label{eq:HStoLaguerre}
\end{alignat}
where $L_n$ is the $n$'th Laguerre polynomial.  

By examination, multiplying the parameters $A,B,r$ of any exponential
Riordan array $\mathcal{HS}=(\HS_{n,k})$ of Hsu--Shiue type by
$c\neq0$ multiplies $\HS_{n,k}$ by $c^{n-k}$ and alters $\HS_n(t)$ to
$c^n\HS_n(t/c)$.  As an application, the $m=n$ case of the
$s$\nobreakdash-ordering conversion formula~(\ref{eq:genconv}) can be
rewritten as
\begin{equation}
  \normord{(\adag a)^n}_s 
  =\normord{\delta^n n!\,L_n(-\adag a/\delta)}_{s'}
  = \normord{\HS_n(\adag a;\,-\delta, \delta,\delta)}_{s'},
\end{equation}
where $\delta=(s-s')/2$.  The arrays
$\mathcal{HS}(-\delta,\delta,\delta)$, $\delta\in\mathbb{C}$, make~up
an `$s$\nobreakdash-shifting' subgroup of~$\mathfrak{R}$ that is
isomorphic to the additive group of~$\mathbb{C}$.

One can see from~(\ref{eq:usedforL}) that the polynomial sequence
$n!\,L_n(-t)$ is the Sheffer sequence $\mathbf{S}[1/(1+\nobreak
  D),\allowbreak D/\allowbreak(1+\nobreak D)]$, or equivalently
$\mathbf{R}[1/(1-\nobreak z), z/\allowbreak (1-\nobreak z)]$; and
$(-1)^nn!\,L_n(t)$ is $\mathbf{S}[1/(1-\nobreak D),\allowbreak
  D/\allowbreak(1-\nobreak D)]$ or $\mathbf{R}[1/(1+\nobreak z),
  z/\allowbreak (1+\nobreak z)]$.  (Cf.~\cite{Roman84}.)  The
just-mentioned $1$\nobreakdash-parameter subgroup of~$\mathfrak{R}$
comprises $\mathbf{S}[1/(1+\nobreak \delta D],\allowbreak
D/(1+\nobreak\delta D)]$ or $\allowbreak
  \mathbf{R}[1/(1-\nobreak\delta z),\allowbreak z/(1-\nobreak\delta
    z)]$, $\delta\in\mathbb{C}$.

\smallskip
An unrelated fact will be needed below: besides
$\HS_{n,k}(0,1,0)=\left\{{n \atop k}\right\}$, one can prove by
induction that $\HS_{n,k}(0,1,1)=\left\{{{n+1} \atop {k+1}}\right\}$.

\section{Exponentiated operators}
\label{sec:expops}

The main result of this paper, on the $s$-ordered representation of
the exponential of a single-mode ``boson string'' of the form $\adag^L
a\adag^R$, will be deduced in Sec.~\ref{sec:main} from a restricted
result on a normal ordering; or equally well, an anti-normal one.  The
normal and anti-normal results, given in
Theorem~\ref{thm:identities34} below, appeared in the Introduction as
eqs.\ (\ref{eq:3}),(\ref{eq:4}).  They will follow from
Theorem~\ref{thm:extract}, a result of the author on the $n$'th power
of $\adag^L a\adag^R$ \cite[Thm.~6.5]{Maier29}, a~short proof of which
is supplied, to make this paper self-contained.

As will be sketched, they also follow from the results of
Lang~\cite{Lang2000,Lang2007} and B{\l}asiak
et~al.\ \cite{Blasiak2005,Blasiak2006} on exponentiated vector fields
and on applications of Sheffer sequences to ordering problems.  That
Fan et~al.~\cite{Fan2006} obtained the $R=\nobreak0$ case of
eq.~(\ref{eq:3}) by the formal IWOP (integration within ordered
products) technique~\cite{Wunsche99} must also be mentioned.

To introduce further the coefficients of Hsu--Shiue and other Sheffer
polynomials, consider the following identities.  Firstly, the two
in~$\mathcal{W}[[\lambda]]$,
\begin{subequations}
\begin{align}
\label{eq:introduce1a}
{\rm e}^{\lambda(\adag a)} &= \normord{\exp\left[ \left({\rm e}^\lambda-1\right)\adag a \right]}_N 
\\
&= \normord{\mathcal{HS}(\adag a,\lambda;\,0,1,0)}_N = \normord{\sum_{k=0}^n\HS_n(\adag a;\,0,1,0)\frac{\lambda^n}{n!}}_N
\label{eq:introduce1b}
\\
\shortintertext{and}
\stepcounter{parentequation}\gdef\theparentequation{\arabic{parentequation}}\setcounter{equation}{0}
\label{eq:introduce2a}
      {\rm e}^{\lambda(a \adag)} &= \normord{{\rm e}^\lambda\exp\left[ \left({\rm e}^\lambda-1\right)a \adag \right]}_N 
\\
&= \normord{\mathcal{HS}(\adag a,\lambda;\,0,1,1)}_N= \normord{\sum_{k=0}^n\HS_n(\adag a;\,0,1,1)\frac{\lambda^n}{n!}}_N.
\label{eq:introduce2b}
\end{align}
\end{subequations}
Here, (\ref{eq:introduce1a}) is the Cahill--Glauber normal ordering
shown in~(\ref{eq:CH1}), and (\ref{eq:introduce2a})~is an immediate
consequence.  The rewritings in~terms of $\mathcal{HS}(t,z)$ come from
the $A\to0$ limit of the EGF formula in Proposition~\ref{prop:HSEGF}.
Secondly, the two in~$\mathcal{W}$,
\begin{align}
  \label{eq:introduce3}
  (\adag a)^n &= \sum_{k=0}^n \left\{{{n} \atop {k}}\right\} \adag^k a^k
  = \sum_{k=0}^n \HS_{n,k}(0,1,0)\,\adag^k a^k
  \\
  \shortintertext{and}
  (a \adag)^n &= \sum_{k=0}^n \left\{{{n+1} \atop {k+1}}\right\} \adag^k a^k 
  = \sum_{k=0}^n \HS_{n,k}(0,1,1)\,\adag^k a^k.
  \label{eq:introduce4}
\end{align}
These are due to Katriel~\cite{Katriel74}, who showed for the first
time how Stirling numbers arise in boson algebra.  Each can also be
proved combinatorially, by applying Wick's theorem and enumerating
contractions~\cite[\S9.3]{Gough2018}.  More generally, one can expand
any polynomial in $\adag a$ or~$a\adag$ in the monomials $\adag^k a^k$
or~$a^k\adag^k$~\cite{VargasMartinez2006}.

Clearly (\ref{eq:introduce1b}),(\ref{eq:introduce2b}) follow from
(\ref{eq:introduce3}),(\ref{eq:introduce4}) by multiplying by
$\lambda^n/n!$ and summing over~$n$.  What is needed for a
generalization to the exponential operator ${\rm e}^{\lambda (\adag^L
  a \adag^R)}$ is a formula for $(\adag^L a \adag^R)^n$.  Duchamp
et~al.\ proved the following.

\begin{proposition}[\cite{Duchamp2004}]
\label{prop:duchampetal}
  If $\omega$ is any ``single annihilator'' word over the alphabet
  $\{a,\adag\}$, i.e., $\omega=\adag^L a \adag^R$, and $e=L+R-1\ge0$,
  then for all\/ $n\ge0$,
  \begin{equation}
    \label{eq:inprop}
    \omega^n 
    = \left(\adag^e\right)^n \left[\sum_{k=0}^n \mathbf{R}[d,h]_{n,k}\,\adag^k a^k\right],
  \end{equation}
  an identity in $\mathcal{W}$, where $\mathbf{R}[d,h]$ is some
  $\omega$-specific exponential Riordan array.
\end{proposition}
What is difficult is finding explicitly the $(L,R)$-dependent series
$d,h\in\mathbb{C}[[z]]$ that define the Riordan array.  Identities
(\ref{eq:introduce2a}),(\ref{eq:introduce2b}) suggest that
$\mathbf{R}[d,h]$ should be of Hsu--Shiue type.  But the derivation of
the identity~(\ref{eq:inprop}) need not be combinatorial, or based on
a generalization of Wick's theorem.

Combinatorial interpretations of the reordered coefficients of such
boson strings as $\adag^{r_n}a^{s_n}\dots
\adag^{r_1}a^{s_1}\in\mathcal{W}$ are known~\cite{Mendez2005}.  Also,
interpretations of the coefficients $\HS_{n,k}(A,B,r)$ for general
integer values of $A,B,r$ are known, such as ones involving the
placing of balls in cells and their compartments, or the drawing of
balls from an urn.  But such combinatorial interpretations tend
at~present to be complicated.  (See, e.g.,
\cite[Chapter~4]{Mansour2016}.)  The following two parametric
identities in~$\mathcal{W}$ are extracted
from~\cite[Theorem~6.5]{Maier29}, and have a non-combinatorial proof.
The first subsumes Katriel's formulae, (\ref{eq:introduce3})
and~(\ref{eq:introduce4}).
\begin{theorem}
  \label{thm:extract}
  For all integer $L,R\ge0$ with $e=L+R-1 \ge0$,
  \begin{align}
    &
    \label{eq:noncombin1}
    \begin{aligned}
      \left(\adag^L a \adag^R\right)^n &= \left(\adag^e\right)^n \left[\sum_{k=0}^n \HS_{n,k}(-e,1,R)\,\adag^k a^k\right]
      \\
      &= \normord{\left(\adag^e\right)^n \HS_{n}(\adag a;\, -e,1,R)}_N
    \end{aligned}
    \\
    \shortintertext{and}
    &
    \label{eq:noncombin2}
    \begin{aligned}
      \left(\adag^L a \adag^R\right)^n &= \left[\sum_{k=0}^n \HS_{n,k}(e,-1,-L)\,a^k \adag^k \right]\left(\adag^e\right)^n 
      \\
      &= \normord{\HS_{n}(a \adag;\,e,-1,-L)\left(\adag^e\right)^n}_A,
    \end{aligned}
\end{align}
  where the $\HS_{n,k}$ coefficients can be computed from
  the summation formula\/ {\rm(\ref{eq:HSsummation})} or
  the recurrence\/~{\rm(\ref{eq:HSrec})}.
\end{theorem}
\begin{proof}
  Representing $a,\adag$ by $D_x,x$, acting on any convenient space of
  functions of~$x$, yields a faithful representation of~$\mathcal{W}$,
  so it suffices to obtain versions of
  (\ref{eq:noncombin1}),(\ref{eq:noncombin2}) with $a,\adag$ replaced
  by $D_x,x$.
  
  Replacing $x$ by~$xD_x$ in the definition~(\ref{eq:subsumes}) of the
  coefficients~$\HS_{n,k}$ yields
  \begin{equation}
    \label{eq:rewritten}
    (xD_x+r)^{\underline{n},A} = \sum_{k=0}^n   HS_{n,k}(A,B,r)\, (xD_x)^{\underline{k},B}.
  \end{equation}
  This can be rewritten with the aid of Boole's lemma
  $(xD_x)^{\underline{n}} = x^nD_x^n$, which has a long
  history~\cite{Mansour2016}.  The more general
  forms of the lemma
  \begin{subequations}
  \begin{align}
    \label{eq:majorrole1a}
    (xD_x)^{\underline n,A} &= (x^A)^n (x^{1-A}D_x)^n,\\
    &= (xD_xx^A)^n (x^{-A})^n
    \label{eq:majorrole1b}
  \end{align}
  \end{subequations}
  (cf.~\cite{MohammadNoori2011}) can be proved by induction, or
  operationally: the left and right-hand sides act identically on
  monomials in~$x$.  The yet more general forms
  \begin{subequations}
  \begin{align}
    \label{eq:majorrole2a}
    (xD_x+r)^{\underline n,A} &= x^{-r}\left[(x^A)^n (x^{1-A}D_x)^n\right]x^r
    =(x^A)^n(x^{1-A-r}D_xx^r)^n,
    \\
    &= x^{-r}\left[ (xD_xx^A)^n (x^{-A})^n \right] x^r
    = (x^{1-r}D_xx^{A+r})^n (x^{-A})^n
    \label{eq:majorrole2b}
    \end{align}
    \end{subequations}
  follow by the similarity transformation $(\cdot)\mapsto
  x^{-r}(\cdot)x^r$.

  To obtain a version of (\ref{eq:noncombin1}) with $a,\adag$ replaced
  by $D_x,x$, rewrite the left-hand side of~(\ref{eq:rewritten}) with
  the aid of~(\ref{eq:majorrole2a}), and the right-hand side with the
  aid of~(\ref{eq:majorrole1a}); and also set $(A,B,r)=(-e,1,R)$.  To
  obtain~(\ref{eq:noncombin2}) similarly, use (\ref{eq:majorrole2b})
  and~(\ref{eq:majorrole1b}), and set
  $(A,B,r)=\allowbreak(e,-1,1-\nobreak L)$.
\end{proof}

\begin{example}
  If $L=R=1$, identities
  (\ref{eq:noncombin1}),(\ref{eq:noncombin2}) of the theorem reduce to
  \begin{gather}
    \label{eq:laguerre1}
      \left(\adag a \adag\right)^n = 
      \left\{
      \begin{aligned}
        &\normord{\adag^n \HS_n(\adag a;\,-1,1,1)}_N  = n!\;\normord{\adag^n L_n(-\adag a)}_N,
        \\
        &\normord{\HS_n(a \adag;\,1,-1,-1) a^n }_A  = (-1)^nn!\;\normord{L_n(a \adag) \adag^n}_A,
      \end{aligned}
      \right.
      \\
      \shortintertext{by (\ref{eq:HStoLaguerre}), so that}
        \label{eq:laguerre2}
        \left(\adag a \adag\right)^n = 
        \left\{
        \begin{aligned}
          &\adag^n \cdot \sum_{k=0}^n \frac{n!^2}{k!^2(n-k)!}\, \adag^k a^k,
          \\
          &\sum_{k=0}^n (-1)^{n-k}\frac{n!^2}{k!^2(n-k)!}\, a^k \adag^k   \cdot \adag^n .
        \end{aligned}
        \right.
\end{gather}
  These are known Laguerre-polynomial
  representations~\cite{Penson2009}, which by the $(s,s')=(\pm1,\mp1)$
  cases of the conversion formula~(\ref{eq:genconv}) can be seen to be
  equivalent to $(\adag a \adag)^n=\adag^na^n\adag^n$, the Viskov
  identity~\cite{Mansour2016}.
\end{example}

Multiplying the identities in Theorem~\ref{thm:extract} by
$\lambda^n/n!$ and summing over~$n$ yields the following identities
in~$\mathcal{W}[[\lambda]]$, with the two explicit expressions coming
from the EGF formula in Proposition~\ref{prop:HSEGF}.  The first
subsumes (\ref{eq:introduce1b}) and~(\ref{eq:introduce2b}).
\begin{theorem}
\label{thm:identities34}
For all integer $L,R\ge0$ with $e=L+R-1 \ge0$, 
$G={\rm e}^{\lambda(\adag^L a \adag^R)}$
has the ordered representations
\begin{align}
&
\begin{aligned}
{\rm e}^{\lambda(\adag^L a \adag^R)}
&=\normord{\mathcal{HS}(\adag a, \lambda\adag^e;\, -e,1,R)}_N\\
&= \normord{(1-e\lambda\adag^e)^{-R/e} \exp\left\{ \left[\bigl(1-e\lambda\adag^e\bigr)^{-1/e}-1\right] \adag a\right\}}_N
\end{aligned}
\label{eq:fullgen}
\\
\shortintertext{and}
&
\begin{aligned}
{\rm e}^{\lambda(\adag^L a \adag^R)}
&=\normord{\mathcal{HS}(a \adag, \lambda\adag^e;\,e,-1,-L)}_A\\
&= \normord{\exp\left\{ - a\adag \left[\bigl(1+e\lambda\adag^e\bigr)^{-1/e}-1\right] \right\}   (1+e\lambda\adag^e)^{-L/e} }_A,
\end{aligned}
\end{align}
with the case $e=0$ handled by taking the formal\/ $e\to0$ limit.
\end{theorem}
\begin{remarkaftertheorem}
The explicit expressions for $G$ contain its
``N\nobreakdash-transform'' and ``A\nobreakdash-transform,'' which are
what are acted upon by the quantizations $\normord{\dots}_N$ and
$\normord{\dots}_A$ and could be denoted by
$F^{(N)}(\alpha,\alpha^*)$, $F^{(A)}(\alpha,\alpha^*)\in\allowbreak
\mathcal{W}_c[[\lambda]]=\allowbreak
\mathbb{C}[\alpha,\alpha^*][[\lambda]]$, as in Sec.~\ref{sec:orderings}.
As was explained, $\alpha,\alpha^*$ are optionally written as
$a,\adag$ inside each~$\normord{\dots}$.
\end{remarkaftertheorem}
\begin{remarkaftertheorem}
The two identities $G=\normord{\dots}_N$ and $G=\normord{\dots}_A$,
with their explict right-hand sides coming from
Proposition~\ref{prop:HSEGF}, imply each other.  To convert either to
the other, (1)~take the adjoint, (2)~replace $a,\adag$ by $\adag,-a$
(an~automorphism of~$\mathcal{W}$, entailing the interchange of
$\normord{\dots}_N$ and $\normord{\dots}_A$), and (3)~negate~$\lambda$
and interchange~$L,R$.
\end{remarkaftertheorem}

The $R=0$ case of~(\ref{eq:fullgen}) was previously derived by
Lang~\cite{Lang2000,Lang2007} and B{\l}asiak
et~al.\ \cite{Blasiak2005}.  They used a representation
of~$\mathcal{W}$ in which $a,\adag$ are represented by~$D_x,x$, and
exponentiated $\adag^L a$ by integrating the vector field~$x^LD_x$.
Also, Penson et~al.~\cite{Penson2009} obtained what amounted to the
$R=1$ case of~(\ref{eq:fullgen}) by exploiting a Dobi{\'{n}}ski
representation of generalized Bell polynomials.

More generally, B{\l}asiak et~al.\ \cite{Blasiak2005} integrated
$q(x)D_x+v(x)$, where $q,v\in\mathbb{C}[[x]]$.  As $\adag^L a\adag^R$
equals $(\adag)^{L+R}a +\nobreak R(\adag)^{L+R-1}$, by substituting
$q(x)=x^{L+R}$ and $v(x)=Rx^{L+R-1}$ into their formulae one can
derive (\ref{eq:fullgen}) in full generality.

B{\l}asiak et~al.\ \cite{Blasiak2006} also used Sheffer polynomials to
normal order an exponentiated version of the operator~$\tilde M$
mentioned above at the end of~\S\ref{sec:Sheffer}.  Up~to the taking
of an adjoint, they derived a formula for ${\rm e}^{\lambda\tilde
  M^\dag}$: they showed that
\begin{equation}
\label{eq:blasiaksheffer}
\begin{split}
&  \exp\left\{
  \lambda\,
  \frac1{f'(\adag)} \left[
    a - \frac{g'(\adag)}{g(\adag)}
    \right]
  \right\}
\\
&  \quad =
  \normord{
    \frac{g(\adag)}{g\left(\bar f(f(\adag)+\lambda)\right)}
    \exp{\Bigl\{\left[\bar f(f(\adag)+\lambda)-\adag\right]a
    \Bigr\}}}_N,
\end{split}
\end{equation}
where $g,f\in\mathbb{C}[[D]]$ are any pair that defines a Sheffer
sequence.  Equation~(\ref{eq:blasiaksheffer}) was not obtained in the
power-series framework of the present paper, as its derivation
employed a Fock space representation and projections onto coherent
states, but it is a valid identity
in~$\mathbb{C}[[\adag]][a][[\lambda]]$.

By examination, \emph{formally} substituting $g(D)=D^{-R}$,
$f(D)=-D^{-e}/e$ into~(\ref{eq:blasiaksheffer})
yields~(\ref{eq:fullgen}).  But this additional derivation of
eq.~(\ref{eq:fullgen}) is not rigorous, as these $g,f$ contain
negative powers of~$D$ and do not belong to~$\mathbb{C}[[D]]$.
At~most, it is an analytic continuation of a rigorous derivation.

\section{A new parametric family and the main result}
\label{sec:main}

The main result (Theorem~\ref{thm:main}) can now be approached.
According to Theorem~\ref{thm:identities34}, the normally ordered
($s=\nobreak-1$) and anti-normally ordered ($s=\nobreak+1$) representations of the
exponential operator ${\rm e}^{\lambda(\adag^L a\adag^R)}$ can be
expressed in~terms of $\mathcal{HS}(t,z;\allowbreak A,B,r)$, the EGF
of the parametric family of exponential Riordan arrays of Hsu--Shiue
type, for which there is an explicit expression.  (Recall
Proposition~\ref{prop:HSEGF}.)  For any $s\in\mathbb{C}$, the
$s$\nobreakdash-ordered representation of ${\rm e}^{\lambda(\adag^L
  a\adag^R)}$ will now be obtained in~terms of a new family of
exponential Riordan arrays, or family of Sheffer polynomial sequences,
which is of independent interest.

\subsection{The new parametric family}

Any member of the Hsu--Shiue family can be viewed as the
$n$\nobreakdash-indexed sequence of row polynomials coming from an
exponential Riordan array $\mathbf{R}[d_\textit{HS}(z),\allowbreak
  h_\textit{HS}(z)]$, for known elementary functions
$d_\textit{HS}(z)$, $h_\textit{HS}(z)$, and its EGF
$\mathcal{HS}(t,z)$ equals $d_\textit{HS}(z) {\rm e}^{t
  h_\textit{HS}(z)}$.  Equivalently, the member is viewed in Sheffer
terms as $\mathbf{S}[g_\textit{HS}(D),\allowbreak f_\textit{HS}(D)]$,
where $[g_\textit{HS},f_\textit{HS}]$ is the inverse in the
exponential Riordan group~$\mathfrak{R}$ of the element
$[d_\textit{HS},h_\textit{HS}]$.  The function $D=h_\textit{HS}(z)$
has compositional inverse $z=f_\textit{HS}(D)$, and $g_\textit{HS}(D)$
equals $1/d_\textit{HS}(z) = 1/d_\textit{HS}(f_\textit{HS}(D))$.

As noted in Sec.~\ref{subsec:HS}, the Sheffer functions
$g_\textit{HS},f_\textit{HS}$ turn~out to be minor variations of the
Riordan functions $d_\textit{HS},h_\textit{HS}$: their parameters
$A,B$ are interchanged and $r$~is negated.  This is a special property
of the Hsu--Shiue family.

When $s\neq \pm1$ and $L+R>1$, the exponential Riordan EGF in~terms of
which the $s$\nobreakdash-ordered representation of ${\rm
  e}^{\lambda(\adag^L a\adag^R)}$ is expressed in
Theorem~\ref{thm:main} will come from a Riordan pair $[d,h]$ in which
$d(z),\allowbreak h(z)$ are not elementary functions of~$z$.  But the
functions in the inverse pair $[g,f]=[d,h]^{-1}$ \emph{will} be
elementary, and one can view $[d,h]$ as $[\bar g,\bar f] \defeq
[g,f]^{-1}$, if necessary, and write the EGF as $\bar g(z){\rm
  e}^{t\bar f(z)}$.  So the sequence of polynomials from which the EGF
is computed can optionally be denoted more briefly and comprehensibly
as $\mathbf{S}[g,f]$ (in~Sheffer terms), than as $\mathbf{R}[d,h]$
(in~exponential Riordan terms).

The following defines a new family of Sheffer polynomial sequences
with this feature, which will appear in Theorem~\ref{thm:main}.  It
extends the Hsu--Shiue family, having the extra parameter~$s$ and the
pair $r,r'$ instead of~$r$.  Why it is called the ``two-point''
Hsu--Shiue family will be clear from Proposition~\ref{prop:arsenal}.

Firstly, define parametric functions $g_T,f_T$ of $D$, or more accurately
parametric Maclaurin series $g_T,f_T\in\mathbb{C}[[D]]$, by
\begin{sizeequation}{\small}
\label{eq:theformulae}
\begin{aligned}
&g_T(D;B,r,r';\,s) \defeq
\left(\frac{1+s}2\right) \left[ \frac{1+(\frac{1-s}2) B\,D}{1-(\frac{1+s}2) B\,D} \right]^{-r/B}\!\!+
\left(\frac{1-s}2\right) \left[ \frac{1-(\frac{1+s}2) B\,D}{1+(\frac{1-s}2) B\,D} \right]^{r'/B},
\\[\jot]
&z=f_T(D;A,B;\,s) \defeq \frac{[1-(\frac{1+s}2) B\,D ]^{-A/B} -1}A + \frac{[1+(\frac{1-s}2) B\,D]^{-A/B}-1}{-A},
\end{aligned}
\end{sizeequation}
with the cases $A=0$, $B=0$ handled by taking limits.  These $g_T,f_T$
have Maclaurin expansions that are series in $\mathbb{C}[[D]]$ of
order $0$ and~$1$ respectively: $g_T=1\cdot D^0 +\nobreak O(D^1)$ and
$f_T=1\cdot D^1 +\nobreak O(D^2)$.

\begin{definition}
\label{def:new}
  The parametric sequence $T_n(t)=T_n(t;\,A,B,r,r';\,s)$ of two-point
  Hsu--Shiue polynomials, denoted by $\mathcal{T}(A,B,r,r';\,s)$, is
  the Sheffer sequence
  $\mathbf{S}[g_T(\cdot;B,r,r';s),\allowbreak f_T(\cdot;A,B;s)]$, with the EGF
  \begin{displaymath}
    \mathcal{T}(t,z;\,A,B,r,r';\,s)
    =\sum_{n=0}^\infty T_n(t;\,A,B,r,r';\,s)  \frac{z^n}{n!} = \bar g_T(z;\,s){\rm e}^{t\bar f_T(z;\,s)}.
  \end{displaymath}
  Here, the pair $[\bar g_T,\bar f_T]$ is the inverse in the
  exponential Riordan group~$\mathfrak{R}$ of the parametric element
  $[g_T,f_T]$ defined in eq.~(\ref{eq:theformulae}).  That is, $D=\bar
  f_T(z;s)$ is inverse to $z=f_T(D;s)$ and can be obtained by series
  reversion, and $\bar g_T=1/g_T(D;s)=1/g_T(\bar f_T(z;s))$.  The
  notations $\mathbf{S}[g_T,f_T]$ and $\mathbf{R}[\bar g_T,\bar f_T]$
  are equivalent.
\end{definition}

It should be clear from the definition of $f_T$
in~(\ref{eq:theformulae}) that for \emph{general}~$s$ and $A,B$, the
inverse $D=\bar f_T(z)$ of $z=f_T(D)$ is not an elementary function.
But:
\begin{proposition}
  \label{prop:arsenal}
  The sequence $\mathcal{T}(A,B,r,r';s)$ of two-point Hsu--Shiue
  polynomials reduces as follows when\/ $s=\pm1$:
  \begin{align*}
    &
    \begin{aligned}
      &{T}_n(t;A,B,r,r';\,-1)
      \\
      &\quad= \mathbf{S}[d_{\textit{HS}}(\cdot;B,-r'),h_{\textit{HS}}(\cdot;B,-A)]_n(t) = \mathbf{R}[d_{\textit{HS}}(\cdot;-A,r'),h_{\textit{HS}}(\cdot;-A,B)]_n(t),
    \end{aligned}
    \\
    \shortintertext{and}
    &
    \begin{aligned}
      &{T}_n(t;A,B,r,r';\,+1)
      \\
      &\quad= \mathbf{S}[d_{\textit{HS}}(\cdot;-B,-r),h_{\textit{HS}}(\cdot;-B,A)]_n(t) = \mathbf{R}[d_{\textit{HS}}(\cdot;A,r),h_{\textit{HS}}(\cdot;A,-B)]_n(t).
    \end{aligned}
  \end{align*}
  That is, individual polynomials reduce according to
  \begin{align*}
    &{T}_n(t;\,A,B,r,r';\,-1) =  {HS}_n(t;\,-A,B,r'),\\
    &{T}_n(t;\,A,B,r,r';\,+1) =  {HS}_n(t;\,A,-B,r),
  \end{align*}
  and their EGF's reduce according to
  \begin{align*}
    \mathcal{T}(t,z;\,A,B,r,r';\,-1) &=  \mathcal{HS}(t,z;\,-A,B,r')
    \\
    &= (1-Az)^{-r'/A} \exp\left\{\frac{t}{B}\left[(1-Az)^{-B/A} - 1\right]\right\}
    \\
    \shortintertext{and}
    \mathcal{T}(t,z;\,A,B,r,r';\,+1) &=  \mathcal{HS}(t,z;\,A,-B,r)
    \\
    &= (1+Az)^{r/A} \exp\left\{\frac{t}{B}\left[1-(1+Az)^{-B/A}\right]\right\}.
  \end{align*}
\end{proposition}
\begin{proof}
  Set $s=-1,+1$ in the formulae (\ref{eq:theformulae}) for $g_T,f_T$,
  and compare them with the formulae (\ref{eq:dhformulae}) for
  $d_\textit{HS},h_\textit{HS}$.  The formula for the parametric
  Hsu--Shiue EGF $\mathcal{HS}(t,z)$ was supplied in
  Proposition~\ref{prop:HSEGF}.
\end{proof}

So, exponential Riordan arrays of the newly defined two-point
Hsu--Shiue type (and likewise their row polynomials and EGF's)
\emph{interpolate} as the parameter $s$ varies over the
$1$\nobreakdash-simplex ${-1\le s\le+1}$ between a pair of
conventional Hsu--Shiue ones, to which they reduce when $s=\pm1$.  One
could also define three-point interpolation, with a parameter vector
that varies over a $2$\nobreakdash-simplex, etc.

\subsection{Main result}

The following theorem states how the $s$-ordered representation of
${\rm e}^{\lambda(\adag^L a \adag^R)}$ can be expressed in~terms of
the two-point Hsu--Shiue EGF $\mathcal{T}(t,z;A,B,r,r';s)$, and
subsumes Theorem~\ref{thm:identities34}, which provided only normal
($s=\nobreak-1$) and anti-normal ($s=\nobreak+1$) orderings.  Examples
are given in Sec.~\ref{sec:examples}.  The corollary to this theorem
(Corollary~\ref{cor:main}) expresses $(\adag^L a \adag^R)^n$ in~terms
of the row polynomials $T_n(t;\allowbreak A,B,r,r';s)$, and subsumes
Theorem~\ref{thm:extract}.

\begin{theorem}
  \label{thm:main}
  For all integer $L,R\ge0$ with $e=L+R-1\ge0$, $G={\rm
    e}^{\lambda(\adag^L a \adag^R)}$ has for all\/ $s\in\mathbb{C}$ the
  $s$\nobreakdash-ordered representation
  \begin{align}
    \label{eq:mainresult}
         {\rm e}^{\lambda(\adag^L a \adag^R)} &= \normord{\mathcal{T}(\adag
      a,\lambda \adag^e;\, e,1,-L,R;\, s)}_s
    \\
    &= \normord{\bar g_T(\lambda \adag^e;\,s) \exp\left[(\adag a)\bar f_T(\lambda \adag^e;\,s)  \right]}_s,
  \end{align}
  where $[\bar g_T,\bar f_T]=[\bar g_T(z;s),\bar f_T(z;s)]$ is the Riordan-group inverse of\/
  $[g_T,f_T]$, defined as in eq.\ {\rm (\ref{eq:theformulae})},
  with parameters $(A,B,r,r';s) = (e,1,-L,R;s)$.  This reduces when
  $s=\pm1$ to
    \begin{alignat*}{2}
      &\normord{(1-e\lambda\adag^e)^{-R/e} \exp\left\{ \left[\bigl(1-e\lambda\adag^e\bigr)^{-1/e}-1\right] \adag a\right\}}_N &\qquad & \textrm{if }s=-1
      \\
      \shortintertext{and}
      &\normord{\exp\left\{ - a\adag \left[\bigl(1+e\lambda\adag^e\bigr)^{-1/e}-1\right] \right\}   (1+e\lambda\adag^e)^{-L/e} }_A &\qquad& \textrm{if }s=+1,   
    \end{alignat*}
  where in both, the $e=0$ case is handled by taking the formal
  $e\to0$ limit.
\end{theorem}
\begin{proof}
  The idea is this: $G={\rm e}^{\lambda(\adag^L a \adag^R)}$ equals
  $\normord{F^{(s)}(\alpha,\alpha^*)}_s$, where
  $F^{(s)}(\alpha,\alpha^*)$ is some function that satisfies the
  diffusion equation~(\ref{eq:unexponentiated}), a~PDE that evolves
  $F^{(s)}$ toward smaller values of~$s$.  (In~this proof, for clarity
  the commuting arguments of $F^{(s)}$ will be written as
  $\alpha,\alpha^*$ and not as~$a,\adag$.  The notational convention
  in quantum optics that dependence on~$\alpha,\alpha^*$ really means
  dependence on a single complex variable
  $\alpha=\allowbreak\alpha_r+\nobreak {\rm i}\alpha_i$ will
  briefly be followed.)

  In~principle, one could start with $F^{(+1)}=F^{(A)}$ and integrate
  the diffusion equation backward in~$s$ toward $F^{(-1)}=F^{(N)}$, as
  one does in quantum optics when smoothing the Glauber--Sudarshan
  $P$\nobreakdash-function ($s=\nobreak+1$) to obtain the Husimi
  $Q$\nobreakdash-function ($s=\nobreak-1$) \cite{Schleich2001,Leonhardt97}.
  The differential operator
  $\partial^2/\partial\alpha\partial\alpha^*$ in the PDE is to be
  interpreted as~$\Delta/4$, with
  $\Delta=\partial^2/\partial\alpha_r^2 +
  \partial^2/\partial\alpha_i^2$ the Laplacian on the complex
  $\alpha$\nobreakdash-plane.

  However, to avoid convergence questions the exponential operator~$G$
  is viewed in this paper as an element of~$\mathcal{W}[[\lambda]]$,
  and $F^{(s)}$~is accordingly an element
  of~$\mathcal{W}_c[[\lambda]]$.  So $F^{(s)}$~is not a single
  function on the complex $\alpha$\nobreakdash-plane (i.e., a~function
  of $\alpha_r,\alpha_i$ or informally of~$\alpha,\alpha^*$), and the
  convolution or smoothing formula~(\ref{eq:smoothing}) cannot easily
  be used.  Instead, $F^{(s)}$ at each value of~$s$ is a power series
  in~$\lambda$, with coefficients that are polynomial
  in~$\alpha,\alpha^*$, and the series is formal: it is never summed.
  The following proof should be read with this in~mind.

  The statement of the theorem incorporates the $s=\pm1$ expressions
  for $F^{(s)}$ from Theorem~\ref{thm:identities34}.  To match their
  functional form, make the Ansatz
  \begin{equation}
    F^{(s)}(\alpha,\alpha^*) = {V}(\alpha^*\alpha, \lambda{\alpha^*}^e;\,s),
  \end{equation}
  where $\lambda\neq0$ and the function $V(t,z;s)$ is to be
  determined.  Changing the independent variables $\alpha,\alpha^*$
  (really $\alpha_r,\alpha_i$) to $t=\alpha^*\alpha$,
  $z=\lambda{\alpha^*}^e$ converts the diffusion
  equation~(\ref{eq:unexponentiated}) to an equation of diffusion type
  on the $(t,z)$\nobreakdash-plane,
  \begin{equation}
    \label{eq:Veqn}
    \left[\frac{\partial}{\partial s} + \mathcal{X}\right]V(t,z;\,s)=0,\qquad
    \mathcal{X} \defeq \frac12  \left(t\frac{\partial^2}{\partial t^2} + \frac{\partial}{\partial t} + ez\,\frac{\partial^2}{\partial t\,\partial z}\right),
  \end{equation}
  irrespective of~$\lambda$; so that formally, 
  \begin{equation}
    \label{eq:integrated}
    V(\cdot,\cdot;\,s) = \exp\left[(s_0-s)\mathcal{X}\right]\, V(\cdot,\cdot;\,s_0), \qquad s<s_0.
  \end{equation}
  There is an obvious
  $z$\nobreakdash-independent solution, namely
  \begin{equation}
    V(t,z;\,s) \propto (s_0-s)^{-1} {\rm e}^{-2t/(s_0-s)},\qquad s<s_0,
  \end{equation}
  for any real~$s_0$.  Depending on $t=\alpha^*\alpha$ but not on
  $z=\lambda{\alpha^*}^e$, for $s_0=+1$ this is in~fact the
  $s$\nobreakdash-ordered quasiprobability density on the complex
  $\alpha$\nobreakdash-plane coming from the density operator of the
  vacuum state~\cite[\S12.4.4]{Schleich2001}, showing how the
  $P$\nobreakdash-function ($s=\nobreak+1$), which is
  a~$\delta$\nobreakdash-function located at $\alpha=0$, broadens into
  a Gaussian as $s$~is decreased and the $Q$\nobreakdash-function
  ($s=\nobreak-1$) is approached.  But it is qualitatively different
  from the expressions $V(t=\alpha^*\alpha,\allowbreak
  z=\lambda{\alpha^*}^e;\allowbreak s=\pm1)$ appearing in
  Theorem~\ref{thm:identities34}.

  To incorporate those functions, make the further Riordan/Sheffer
  Ansatz
  \begin{equation}
    V(t,z;\,s)=\mathbf{R}[\bar g(\cdot;s),\bar f(\cdot;s)](t,z)=\mathbf{S}[g(\cdot;s),f(\cdot;s)](t,z),
  \end{equation}
  which says that
  \begin{equation}
    \label{eq:tobesubst}
    V(t,z;\,s) = \bar g(z;\,s)\,{\rm e}^{t\bar f(z;\,s)},
  \end{equation}
  where as usual, $\bar g(z;s)=1\cdot z^0+O(z^1)$ and $D=\bar f(z;s) =
  1\cdot z^1 + O(z^2)$, with $g(D;s)=1\cdot D^0+O(D^1)$ and $z=f(D;s)
  = 1\cdot D^1 + O(D^2)$.  Substituting (\ref{eq:tobesubst})
  into~(\ref{eq:Veqn}) and separating the terms proportional to $t^0$
  and~$t ^ 1$ gives
  \begin{subequations}
    \label{eq:pairfgbar}
    \begin{align}
      & 2\bar g_s + ez\bar f\bar g_z + (\bar f + ez\bar f_z) g = 0,\\
      &2\bar f_s + ez\bar f \bar f_z + \bar f ^2 = 0,
    \end{align}
  \end{subequations}
  where subscripts denote partial derivatives.  These are coupled
  nonlinear PDEs for $\bar g(z;s)$, $\bar f(z;s)$.  The two-point
  boundary conditions
  \begin{subequations}
    \begin{alignat}{2}
    \label{eqs:gbarbc}
      &\bar g(z;\, s=-1) = (1-ez)^{-R / e} ,&\qquad&  \bar g(z;\, s=+1) = (1+ez) ^ {-L/e},
    \\
    \label{eqs:fbarbc}
      &\bar f(z;\, s=-1) = (1-ez) ^ {-1/ e} - 1,&\qquad& \bar f(z;\, s=+1) = 1 - ( 1+  ez ) ^ {-1/ e}
    \end{alignat}
  \end{subequations}
  come from Theorem~\ref{thm:identities34}.

  The function (or more accurately formal infinite series) $z=f(D;s)$
  is the inverse of $D=\bar f(z;s)$, computed by series reversion.
  With its aid, the independent variables $z,s$
  in~(\ref{eq:pairfgbar}) can be replaced by~$D,s$, and the dependent
  variables $\bar g,\bar f$ by $g,f$, where $g= 1 / ( \bar g \circ
  f)$, i.e., $\bar g = 1/(g \circ \bar f)$.  This yields the PDEs
  \begin{subequations}
    \label{eq:dualpdes}
    \begin{align}
      & -2f_D\, g_s + D ^ 2 f_D\, g_D + (Df_D + ef) g = 0,\label{eq:dualgeqn}
      \\
      &-2 f_s + D^ 2f_D + eDf = 0, \label{eq:dualfeqn}
    \end{align}
  \end{subequations}
  which are satisfied by $f(D;s)$ and $g(D;s)$.  The two-point
  boundary conditions on these dual functions follow from
  (\ref{eqs:gbarbc}),(\ref{eqs:fbarbc}) and are
  \begin{subequations}
    \begin{sizealignat}{\small}{2}
      &g(D;\, s=-1) = (1+D)^{-R}, &\quad& g(D;\, s=+1) = (1-D)^{-L},     
      \label{eqs:gbcs}
      \\
      &f(D;\, s=-1) = \left[1-(1+D)^ {-e}\right] /e, &\quad& f(D;\, s=+1) = \left[(1-D) ^ {-e} - 1\right] / e. 
      \label{eqs:fbcs}
    \end{sizealignat}
  \end{subequations}
  The ``dualized'' PDEs (\ref{eq:dualpdes}) are linear and of first order, and
  can be solved by Lagrange's method of characteristics.

  The solution of (\ref{eq:dualfeqn}) is
  \begin{equation}
    z = f(D;\, s) = D ^ {-e} F(\eta ),
  \end{equation}
  a similarity solution in which $\eta \defeq s-2/D$ is the similarity
  variable and the function $F(\eta )$ is arbitrary.  To match the
  boundary conditions~(\ref{eqs:fbcs}), choose
  \begin{gather}
    F(\eta ) = \frac{\left(\frac{-1-\eta }2\right)^{ -e}  -\left(\frac{1-\eta }{2}\right)^{-e}}e,
    \\
    \shortintertext{so that}
    \label{eq:finalf}
    z=f(D;\,s) = 
    \frac{ \left[1-(\frac{1+s}{2})D\right]^{-e} 
      -  {\left[1+(\frac{1-s}2)D\right]^{-e}}}e,
  \end{gather}
  with expansion $f(D;s)=1\cdot D^1 + O(D^2)$.  With this choice of
  $f(D;\,s)$, the PDE~(\ref{eq:dualgeqn}) is solved in the same way.
  It has the similarity solution
  \begin{equation}
    g(D;\,s) = 
    \left[
      \left(\frac{1+s}2\right) + \left(\frac{1-s}2\right) \left( \frac{\eta +1}{\eta -1}\right)^{e+1}
    \right]G(\eta ),
  \end{equation}
  where the function $G(\eta )$ is arbitrary.
To match the boundary conditions~(\ref{eqs:gbcs}), choose
\begin{equation}
  G(\eta) = \left( \frac{\eta-1}{\eta+1} \right)^L,
\end{equation}
so that
\begin{equation}
\label{eq:finalg}
  g(D;\,s)=
\left(\frac{1+s}2\right) \left[ \frac{1+(\frac{1-s}2) D}{1-(\frac{1+s}2) D} \right]^{L}+
\left(\frac{1-s}2\right) \left[ \frac{1-(\frac{1+s}2) D}{1+(\frac{1-s}2) D} \right]^{R},
\end{equation}
with $g(D;s)=1\cdot D^0 + O(D^1)$.  Comparing
(\ref{eq:finalf}),(\ref{eq:finalg}) with the definitions of $f_T,g_T$
given in~(\ref{eq:theformulae}) reveals that
\begin{equation}
  f(D;\,s) = f_T(D;e,1;\,s),
  \qquad
  g(D;\,s) = g_T(D;1,-L,R;\,s).
\end{equation}
Therefore
\begin{equation}
  V(t,z;\,s) = \mathcal{T}(t,z;\, e,1,-L,R;\,s),
\end{equation}
which is the substance of the theorem.
\end{proof}

\begin{remark*}
  Wilcox~\cite{Wilcox70} solved a PDE similar to~(\ref{eq:Veqn}) by
  normal ordering an exponentiated second-order differential operator
  analogous to the $\exp\left[(s_0-\nobreak s)\mathcal{X}\right]$ in
  eq.~(\ref{eq:integrated}).  In his paper the PDE came first, but in
  the just-concluded proof of Theorem~\ref{thm:main} the normal
  ordering (actually $s$\nobreakdash-ordering) came first, and
  determined the PDE to be solved.
\end{remark*}

\begin{corollary}
\label{cor:main}
  For all integer $L,R\ge0$ with $e=L+R-1\ge0$, 
  \begin{displaymath}
    (\adag^L a \adag^R)^n = \normord{\left(\adag^e\right)^nT_n(\adag a;e,1,-L,R;s)}_s,
  \end{displaymath}
  where $T_n(t;e,1,-L,R;s)$ is the case $(A,B,r,r')=(e,1,-L,R)$ of the
  $n\!$'th parametric Sheffer polynomial\/ $\mathbf{S}[g_T,f_T]_n(t) =
  \mathbf{R}[\bar g_T,\bar f_T]_n(t)$, with $g_T(D;s)$ and
  $z=f_T(D;s)$ given in
  Definition~\ref{def:new}, i.e., in eq.\/~(\ref{eq:theformulae}).
\end{corollary}
\begin{proof}
  Expand each side of the identity~(\ref{eq:mainresult}) in
  Theorem~\ref{thm:main} as a formal power series in~$\lambda$, and
  equate the coefficients of~$\lambda^n$ on the two sides.
\end{proof}

The $s=-1$ (normal ordering) case of the formula in this corollary is
an illustration of Prop.~\ref{prop:duchampetal}, the result of
Duchamp {et~al.}

The following is the Weyl ordering ($s=\nobreak0$) case of the main result
Theorem~\ref{thm:main}, which is the one that can be expressed most
simply.  The functions $\bar g(z)$, $D=\bar f(z)$ come from $g(D)$,
$z=f(D)$, which in~turn come from the two-point functions $g_T(D;s)$,
$z=f_T(D;s)$ defined in~(\ref{eq:theformulae}), by setting $s=\nobreak0$.

\begin{theorem}
  \label{thm:tired}
  For all integer $L,R\ge0$ with $e=L+R-1\ge0$, one has the Weyl
  ordering
  \begin{displaymath}
    {\rm e}^{\lambda(\adag^L a \adag^R)} = \normord{\bar g(\lambda \adag^e)\,
    \exp \left[(\adag a)\,\bar f(\lambda \adag^e)\right]}_W,
  \end{displaymath}
  an identity in\/ $\mathcal{W}[[\lambda]]$ in which $D=\bar f(z)$ is
  the compositional inverse of
  \begin{displaymath}
    z=f(D)=\frac{(1-\frac12D)^{-e} - 1}{e} + \frac{(1+\frac12D)^{-e}-1}{-e}
  \end{displaymath}
  with expansion $D=\bar f(z)=1\cdot z^1 + O(z^2)$, and $\bar
  g(z)=1/g(D)=1/g(\bar f(z))$ where
  \begin{displaymath}
    g(D) = \frac12
    \left[
      \frac{1+\frac12D}{1-\frac12D}
      \right]^L
    \!+ \frac12
    \left[
      \frac{1-\frac12D}{1+\frac12D}
      \right]^R,
  \end{displaymath}
  with expansion\/ $\bar g(z)=1\cdot z^0 + O(z^1)$
  The case $e=0$ is handled by taking the formal\/ $e\to0$ limit.
\end{theorem}

\section{Examples}
\label{sec:examples}

The main result (Theorem~\ref{thm:main}) will now be applied, to
produce several new $s$\nobreakdash-ordered representations of
exponential boson operators ${\rm e}^{\lambda( \adag^L
  a\adag^R)}\in\mathcal{W}[[\lambda]]$, and of $(\adag^L a
\adag^R)^n\in\nobreak\mathcal{W}$.

According to the theorem, for any $(L,R)$ one needs the EGF of a
certain family of two-point Hsu--Shiue polynomials (of~Sheffer type),
denoted by $\mathcal{T}(t,z;\allowbreak A,B,r,r';s)$.  The EGF equals
$\bar g_T(z){\rm e}^{t\bar f_T(z)}$, where $[\bar g_T,\bar f_T]$ is
the inverse in the exponential Riordan group~$\mathfrak{R}$ of the
parametric $s$\nobreakdash-dependent element $[g_T,f_T]$ defined in
eq.~(\ref{eq:theformulae}).  One simply substitutes
$(A,B,r,r')=(e,1,-L,R)$, where $e\defeq\allowbreak L+\nobreak
R-\nobreak 1$, and sets $(t,z)=(\adag a,\lambda\adag^e)$ (which really
means $(t,z)=(\alpha^* \alpha,\lambda{\alpha^*}^e)$), to obtain the
desired $s$-ordered representation.

The only possible stumbling block is the computation of the inverse
group element (the function $D=\bar f_T(z;s)$ is the compositional
inverse of $z=f_T(D;s)$).  This potential difficulty has already
appeared in the statement of Theorem~\ref{thm:tired}, without comment.
The following subsections treat the cases when $e=0,1,2$.

\subsection{The case $e=0$}
\label{subsec:e0}

This includes the operators ${\rm e}^{\lambda(\adag a)}$, ${\rm
  e}^{\lambda(a \adag)}$, with $(L,R)=(1,0)$, $(0,1)$.  Substituting
$(A,B,r,r')=\allowbreak (e,1,-L,R) =\allowbreak (0,1,-1,0)$
into~(\ref{eq:theformulae}) gives
\begin{equation}
\label{eq:82}
  g_T(D;\,s) = \frac2{2-D-sD},\qquad
  z=f_T(D;\,s) = \ln\left(\frac{2+D-sD}{2-D-sD}
  \right).
\end{equation}
(In deriving the latter, substituting $A=0$ means taking the $A\to0$
limit.)  The inverse $D=\bar f_T(z;s)$ and $\bar g_T(z;s) =
1/g_T(D=\bar f_T(z;s);s)$ become
\begin{equation}
  \bar g_T(z;\,s) = \frac2{1+{\rm e}^z - s(1-{\rm e}^z)},\qquad
  D=\bar f_T(z;\,s) = \frac{2({\rm e}^z - 1)}{1+{\rm e}^z - s(1-{\rm e}^z)}.
\end{equation}
Substituting into $\bar g_T(z;s){\rm e}^{t\bar f_T(z;s)}$ and setting
$(t,z) = (\adag a, \lambda \adag^e) =\allowbreak (\adag a,\lambda)$
yields
\begin{equation}
\label{eq:CH22}
  {\rm e}^{\lambda(\adag a)}
  =
  \normord{
  \frac{2}{1+{\rm e}^\lambda - s(1-{\rm e}^\lambda)}
        \exp \left [ \frac{ 2({\rm e}^\lambda-1)}{1+{\rm e}^\lambda-s(1-{\rm e}^\lambda)} \adag a \right] }_s,
\end{equation}
which is precisely the original Cahill--Glauber
$s$\nobreakdash-ordering result, eq.~(\ref{eq:CH2}).  The functions
$g_T,f_T$ given in~(\ref{eq:82}) appear in its inverse,
eq.~(\ref{eq:invCH2}).  The case $(L,R)=(0,1)$ is similar to the case
$(L,R)=(1,0)$.

\subsection{The case $e=1$}
\label{subsec:e1}

This includes the operators ${\rm e}^{\lambda(\adag^2 a)}$, ${\rm
  e}^{\lambda(\adag a\adag)}$, ${\rm e}^{\lambda(a\adag^2 )}$, with
$(L,R)=(2,0)$, $(1,1)$, $(0,2)$.  Substituting $(A,B)=(e,1)=(1,1)$
into~(\ref{eq:theformulae}) gives
\begin{align}
  &
  g_T(D;\,s) =
  \left\{
  \begin{alignedat}{2}
    &\frac{4+D^2-s^2D^2}{(2-D-sD)^2}, &\qquad& (L,R)=(2,0);
    \\
    &\frac{4+D^2-s^2D^2}{(2-sD)^2 - D^2}, &\qquad& (L,R)=(1,1);
    \\
    &\frac{4+D^2-s^2D^2}{(2+D-sD)^2}, &\qquad& (L,R)=(0,2);
  \end{alignedat}
  \right.
  \\
  \shortintertext{and}
  &
  z=f_T(D;\,s)= \frac{4D}{4-D^2}.
\end{align}
Hence,
\begin{align}
\label{eq:hence1}
  &
  \bar g_T(z;\,s) =
  \left\{
  \begin{alignedat}{2}
    &\frac{1+s}{1+2sz+z^2 + (z+s)\sqrt{1+2sz+z^2}}, &\qquad& (L,R)=(2,0);
    \\
    &\frac1{\sqrt{1+2sz+z^2}}, &\qquad& (L,R)=(1,1);
    \\
    &\frac{1-s}{1+2sz+z^2 - (z+s)\sqrt{1+2sz+z^2}}, &\qquad& (L,R)=(0,2);
  \end{alignedat}
  \right.
  \\
  \shortintertext{and}
  &
  D= \bar f_T(z;\,s)= \frac{2z}{1+sz+\sqrt{1+2sz+z^2}}.
  \label{eq:hence2}
\end{align}
Substituting into $\bar g_T(z;s){\rm e}^{t\bar f_T(z;s)}$, with $(t,z)
= (\adag a,\lambda\adag^e) = (\adag a,\lambda\adag)$, yields
\begin{theorem}
  \label{thm:3reps}
  One has the $s$-ordered representations
  \begin{gather}
    \begin{aligned}
      {\rm e}^{\lambda (\adag^2 a)} &=     
      \mathopen{:}\,
      \frac{1+s}{1+2s\lambda\adag+\lambda^2 \adag^2 + (\lambda\adag+s)\sqrt{1+2s\lambda\adag+\lambda^2 \adag^2}}
      \times\mathcal{E}
      \,\mathclose{:}_s,
      \\
        {\rm e}^{\lambda (\adag a\adag)} &=
        \mathopen{:}\,
        \frac1{\sqrt{1+2s\lambda\adag+\lambda^2 \adag^2}}\times
        \mathcal{E}
        \,\mathclose{:}_s,
        \\
          {\rm e}^{\lambda (a \adag^2)} &= 
          \mathopen{:}\,
          \frac{1-s}{1+2s\lambda\adag+\lambda^2 \adag^2 - (\lambda\adag+s)\sqrt{1+2s\lambda\adag+\lambda^2 \adag^2}}
          \times\mathcal{E}\,
          \mathclose{:}_s,
    \end{aligned}
    \\
    \shortintertext{where}
    \begin{aligned}
      \mathcal{E}\defeq \exp\left[(\adag a)\,\frac{2\lambda\adag}{1+s\lambda\adag+\sqrt{1+2s\lambda\adag+\lambda^2 \adag^2}}\right],
    \end{aligned}
  \end{gather}
  for all\/ $s\in\mathbb{C}$.
\end{theorem}
When $s=\nobreak-1$ and $s=\nobreak+1$, each of these three ordered
representations reduces to a normal ordered and an anti-normal ordered
one, which are the ones provided by Theorem~\ref{thm:identities34}.
(To set $s=\nobreak-1$ in the first or $s=\nobreak+1$ in the third, a
limit must be taken.)  But, the three Weyl-ordered representations
obtained by setting $s=\nobreak0$ seem to be new to the literature.

From such identities as those of Theorem~\ref{thm:3reps}, it is
possible to derive explicitly $s$\nobreakdash-ordered series
expansions of $(\adag^2 a)^n$, $(\adag a\adag)^n$, $(a \adag^2)^n$; at
least, when $s=-1,0,+1$.  For the case of $(\adag a\adag)^n$ when
$s=\pm1$, this has already been done: see the Laguerre-polynomial
representations in eqs.\ (\ref{eq:laguerre1}),(\ref{eq:laguerre2})
above, each term in which is explicitly normal, resp.\ anti-normal
ordered.  The expansion of $(\adag a\adag)^n$ when $s=\nobreak0$ is of
particular interest: it expresses $(\adag a\adag)^n$ as the symmetric
(Weyl) quantization of a classical function on the complex
$\alpha$\nobreakdash-plane.
\begin{proposition}
\label{prop:lastprop}
  One has the Weyl-ordered\/ {\rm($s=\nobreak0$)} representation
  \begin{displaymath}
    (\adag a\adag)^n=
    \sum_{{0\le k\le n}\atop{\text{$n-k$ even}}} (-1)^{(n-k)/2}\, 2^{-(n-k)} \binom{n}{(n-k)/2}\,\frac{n!}{k!}\,\normord{\adag^n (\adag a)^k}_W.
  \end{displaymath}
\end{proposition}
\begin{remark*}
  The coefficients in this expansion, unlike those in
  (\ref{eq:laguerre1}),(\ref{eq:laguerre2}), are rational rather than
  integer-valued.  
\end{remark*}
\begin{proof}
  Consider the $(L,R)=(1,1)$ case of Corollary~\ref{cor:main} when
  $s=\nobreak0$, or more concretely the $(L,R)=(1,1)$ case of
  Theorem~\ref{thm:tired}, or the second identity in
  Theorem~\ref{thm:3reps} when $s=\nobreak0$.  They all imply that
  \begin{equation}
    (\adag a\adag)^n = \normord{\adag^n\,\mathbf{R}[\bar g_T(z;\,0),\bar f_T(z;\,0)]_n(t=\adag a)}_W,
  \end{equation}
  where $\bar g_T(z;0)$ comes from the $(L,R)=(1,1)$ case
  of~(\ref{eq:hence1}), and $\bar f_T(z;0)$ comes similarly
  from~(\ref{eq:hence2}).  That is,
  \begin{equation}
    (\adag a \adag)^n = \normord{\adag^n\, \mathbf{R}\left[ \frac1{\sqrt{1+z^2}}, \frac{2z}{1+\sqrt{1+z^2}}\right]_n\!(t=\adag a)}_W.
  \end{equation}
  Up to normalization, the Riordan array $\mathbf{R}[d(z),h(z)]$
  appearing here is familiar: versions appear as \oeisseqnum{A108044}
  and \oeisseqnum{A120616} in the OEIS~\cite{OEIS2025}.  The
  expression
  \begin{equation}
    \mathbf{R}_o\left[\frac1{\sqrt{1+4z^2}}, \frac{2z}{1+\sqrt{1+4z^2}}\right]_{n,k}\!
    = (-1)^{(n-k)/2}\binom{n}{(n-k)/2}
  \end{equation}
  for its matrix elements is known (the binomial coefficient is zero
  if $n-\nobreak k$ is odd).  Here, $\mathbf{R}_o$ in distinction to
  $\mathbf{R}$ signifies an \emph{ordinary} (non-exponential) Riordan
  array, meaning that in the definition of
  $\mathbf{R}_o[d(z),h(z)]_{n,k}$ the $n!/k!$ factor in the
  definition~(\ref{eq:expRdef}) used in this paper is absent.
  Altering the normalization and restoring the $n!/k!$ factor yields
  the claimed formula for $(\adag a \adag)^n$.
\end{proof}

\subsection{The case $e=2$}
\label{subsec:e2}

This includes the operators ${\rm e}^{\lambda(\adag^3 a)}$, ${\rm
  e}^{\lambda(\adag^2 a\adag)}$, ${\rm e}^{\lambda(\adag a\adag^2)}$,
${\rm e}^{\lambda( a\adag^3)}$, with $(L,R)=(3,0)$, $(2,1)$, $(1,2)$,
$(0,3)$.  As with $e=0,1$, one substitutes $(A,B)=(e,1)=(2,1)$ and
$L,R$ into the definitions in eq.~(\ref{eq:theformulae}) of the
functions $g_T(D;s)$ and $z=f_T(D;s)$, and computes $[\bar
  g_T(z;s),\allowbreak\bar f_T(z;s)]$ by inverting $[g_T,f_T]$ in the
exponential Riordan group.  The $s$\nobreakdash-ordered representation
of the operator is then $\bar g_T(\lambda \adag^e;s)\times\allowbreak
\exp\left[(\adag a)\bar f_T(\lambda\adag^e;s)\right]$.

But for each of the four operators, the representation is quite
complicated.  The stumbling block is that according
to~(\ref{eq:theformulae}), $z=f_T(D;s)$ is given by
\begin{equation}
  z=f_T(D;s) = \frac{8D(2-sD)}{\left[(2-sD)^2 - D^2\right]^2},
\end{equation}
and solving for $D=\bar f_T(z;s)$ entails solving the quartic
\begin{sizeequation}{\small}
\label{eq:thequartic}
(1-s^2)^2z\,D^4 + 8s(1-s^2)z\,D^3 - 8\left[(1-3s^2)z-s\right]D^2 - 16(1+2sz)\,D + 16z=0.
\end{sizeequation}
The solution $D=D(z;s)$ satisfying $D=1\cdot z^1 +\nobreak O(z^2)$ can
be expressed in radicals by employing the quartic formula, but owing
to its complexity it is not given here.  If $s=\pm1$, the
quartic~(\ref{eq:thequartic}) degenerates to a more easily solved
quadratic, which can be viewed as the reason why the normal and
anti-normal orderings of each of the four operators, which are
supplied by Theorem~\ref{thm:identities34}, are relatively simple.
But for general~$s$, and even when $s=\nobreak0$ (the Weyl ordering
case), there seems to be no~simplification.

\smallskip
Whether it is possible to express $s$-ordered representations of
operators ${\rm e}^{\adag^L a \adag^R}$ with $e=L+\nobreak R-\nobreak
1>2$ in~terms of radicals, for general or even for specialized~$s$, is a
question requiring further study.



\end{document}